\documentclass[twocolumn,10pt]{IEEEtran}

\usepackage{multirow}
\usepackage{graphicx}
\usepackage{color}
\usepackage{epstopdf}
\usepackage{caption}
\usepackage{subcaption}
\usepackage{amsmath}
\usepackage{bm}
\usepackage{pgfplots}
\usepackage{footnote}
\usepackage{stfloats}
\usepackage{placeins}
\usepackage{color,soul}
\usepackage{dsfont}
\usepackage{amssymb}
\usepackage[algoruled]{algorithm2e}
\usepackage{setspace}
\usepackage{pbox}
\usepackage{enumerate}
\usepackage{mathabx}
\usepackage{etextools}
\usepackage{mathtools}
\usepackage[hidelinks]{hyperref}

\usepackage{amsthm}

\theoremstyle{plain}
\newtheorem{thm}{Theorem} % reset theorem numbering for each chapter

\mathtoolsset{showonlyrefs}

\newtheorem{lemma}{Lemma}

\theoremstyle{definition}
\newtheorem{defn}[thm]{Definition} % definition numbers are dependent on theorem numbers
\newtheorem{exmp}[thm]{Example} % same for example numbers

\newlength
\figureheight
\newlength
\figurewidth

\def\x{{\mathbf x}}
\def\X{{\mathbf X}}

\def\d{{\mathbf d}}

\def\uh{ \overline{\mathbf{u}}}
\def\u{ \mathbf{u}}
\def\r{ \mathbf{r}}
\def\xi{\x_i}

\def\Y{{\mathbf Y}}

\def\u{{\mathbf u}}

\def\S{{\mathbf S}}

\def\D{{\mathbf D}}

\def\M{{\mathbf Z}}
\def\P{{\mathbf P}}

\def\alfa{{\boldsymbol \alpha}}
\def\gama{{\boldsymbol \gamma}}
\def\Gama{{\boldsymbol \Gamma}}
\def\Delt{{\boldsymbol \Delta}}

\def\O{{\boldsymbol \Omega}}

\def\Poi{{P_{0,\infty}}}
\def\Loi{{\ell_{0,\infty}}}
\def\Poie{{P_{0,\infty}^\epsilon}}

\def\Lamda{{\boldsymbol \Lambda}}

\def\R{{\mathbf R}}

\def\E{{\mathbf E}}
\def\DT{\D_\mathcal{T}}
\def\DTB{\D_{\overline{\mathcal{T}}}}

\DeclarePairedDelimiterX\set[1]\lbrace\rbrace{#1}

\begin{document}

\title{Working Locally Thinking Globally - Part II: Stability and Algorithms for Convolutional Sparse Coding}
\author{Vardan~Papyan*,
	Jeremias~Sulam*,
	Michael~Elad

\thanks{*The authors contributed equally to this work. 
	
All authors are with the Computer Science Department, the Technion - Israel Institute of Technology.}}

\maketitle

\begin{abstract}
The convolutional sparse model has recently gained increasing attention in the signal and image processing communities, and several methods have been proposed for solving the pursuit problem emerging from it -- in particular its convex relaxation, Basis Pursuit.
In the first of this two-part work, we have provided a theoretical back-bone for this model, providing guarantees for the uniqueness of the sparsest solution and for the success of pursuit algorithms by introducing the notion of stripe sparsity and other related measures.
Herein, we extend the analysis to a noisy regime, thereby considering signal perturbations and model deviations.
We address questions of stability of the sparsest solutions and the success of pursuit algorithms, both greedy and convex.
Classical definitions such as the RIP are generalized to the convolutional model, and existing notions such as the ERC are connected to our setting.
On the algorithmic side, we demonstrate how to solve the global pursuit problem by using simple local processing, thus offering a first of its kind bridge between global modeling of signals and their patch-based local treatment.
\end{abstract}

\vspace{-0.1cm}

\vspace{-0.2cm}
\section{Introduction}
The convolutional sparse coding model has enjoyed of growing popularity in recent years, overcoming some of the limitations of traditional sparse representations \cite{Bruckstein2009}. This model assumes that a global signal can be factorized into the multiplication of a dictionary, which is assumed to be a concatenation of Circulant and banded matrices, and a sparse vector. This, in turn, results in a global model which admits a shift invariant local structure -- a common assumption in signal and image processing.

Although several works have proposed efficient algorithms to solve the corresponding pursuit \cite{Bristow2013,Wohlberg2016,Bristow2014,Heide2015,Kong2014}, very little is known about the theoretical guarantees for the success of these methods, or their connection to classical sparsity-based models. Typical results in the sparse representation literature (see \cite{Elad_Book} for a thorough review) are given in terms of the total number of non-zeros in the representation vector and offer weak conditions in this setting as they disregard the intrinsic architecture of the model. In particular, these results become meaningless when the length of the signal grows, as in the case of natural images.

In the first of this two-part work \cite{Papyan2016_1}, we have presented a detailed study of this model, establishing guarantees for the uniqueness of the sparsest solution for the convolutional problem, where sparsity is measured in terms of a novel quantity, that of the $\Loi$ norm, which considers sparsity in overlapping stripes.
%Moreover, we have ensured the success of a greedy algorithm -- the Orthogonal Matching Pursuit (OMP) \cite{Pati1993a,Chen1989} -- and the Basis Pursuit (BP) \cite{Chen1989} convex relaxation in recovering such a solution in a noiseless regime.
Moreover, we have ensured the success of prominent algorithms in the sparse literature -- in particular, the Orthogonal Matching Pursuit (OMP) and the Basis Pursuit (BP) -- in recovering such a solution in the noiseless case.

These results have shown the importance and benefits of performing a localized analysis of the global convolutional setting. However, these are not directly applicable to real signal scenarios, and in particular to recently developed algorithms \cite{Bristow2014}, as they assume ideal constructions. In this sequel, we undertake the study of a noisy regime, allowing for measurement errors and model deviations. To this end, we generalize and tie past theoretical constructions, such as the Restricted Isometry Property (RIP) \cite{Candes2005} and the Exact Recovery Condition (ERC) \cite{Tropp2004}, to the convolutional framework, proving the stability of this model in this case as well. Furthermore, we show that the solutions found by OMP and BP remain in the vicinity of the underlying sparse vector, thus providing theoretical guarantees for the above methods.

From a practical point of view, despite many works having addressed the convolutional sparse coding problem in a variety of applications \cite{Morup2008,Zhu2015,Zeiler2010,Kavukcuoglu2010,Bristow2014,Huang2015}, a connection to local patch-based models is still missing, with the exception of the recent work in \cite{batenkov2017global}. This is somewhat surprising, as this local treatment has been demonstrated to be efficient in common signal and image processing tasks \cite{Dabov2006,Zoran2011,Romano2015}. In this paper we propose to bridge this gap by solving the global pursuit using solely local operations relying on a shift invariant local model, while preserving the optimality of the overall global pursuit.

In the following section, we begin by reviewing traditional stability results in classic sparse representations theory, before describing the convolutional sparse model and summarizing our main results from part I. The main problem handled in this paper is formally defined in Section \ref{Sec:Global2Local}, and then analyzed in Section \ref{Sec:TheoreticalAnalysis}. The practical aspects of the global pursuit by means of local processing is delineated in Section \ref{Sec:FromGlobal2LocalProcessing}, where two algorithms are proposed. We finally conclude this work in Section \ref{Sec:Conclusions}, proposing exciting future directions.

\section{Preliminaries}
\label{Sec:Preliminaries}
\subsection{The Global Sparse Model}
In the sparse representation model one assumes that a signal $\X \in \mathbb{R}^N$ can be decomposed as $\X = \D \Gama$, where $\D \in \mathbb{R}^{N\times M}$, $\Gama \in \mathbb{R}^M$ and $\|\Gama\|_0 \ll N$. Given such a signal, finding its sparsest representation is known as Sparse Coding, and it attempts to solve the constrained $P_0$ problem:
\begin{equation}
(P_0): \quad \underset{\Gama}{\min} \ \|\Gama\|_0 \ \text{ s.t. } \  \D\Gama = \X. 
\label{Eq:P0problem}
\end{equation}
When dealing with natural signals, the $P_0$ problem is often relaxed to consider model deviations as well as measurement noise. In this set-up one assumes $\Y = \D\Gama + \E$, where $\E$ is a nuisance vector of bounded energy, $\| \E\|_2 \leq \epsilon$. The corresponding recovery problem can then be stated as follows:
\begin{equation}
(P_0^\epsilon): \quad \underset{\Gama}{\min} \ \|\Gama\|_0 \ \text{ s.t. } \|\D\Gama - \Y\|_2 \leq \epsilon.
\end{equation}
Unlike the noiseless case, given a solution to the above problem, one can not claim its uniqueness in solving the $P_0^\epsilon$ problem but instead can guarantee that it will be close enough to the true vector $\Gama$ that generated the signal $\Y$.
This kind of stability results have been derived in recent years by leveraging the Restricted Isometry Property (RIP) \cite{Candes2005}. A matrix $\D$ is said to have a k-RIP with constant $\delta_k$ if this is the smallest quantity such that
\begin{equation}
(1-\delta_k) \| \Gama \|_2^2 \leq \| \D\Gama \|_2^2 \leq (1+\delta_k) \| \Gama \|_2^2,
\end{equation}
for every $\Gama$ satisfying $\| \Gama \|_0=k$. Based on this property, it was shown that assuming $\Gama$ is sparse enough, the distance between $\Gama$ and all other solutions to the $P_0^\epsilon$ problem is bounded \cite{Elad_Book}. Similar stability claims can be formulated in terms of the mutual coherence also, by exploiting its relationship with the RIP property \cite{Elad_Book}. More on these results is brought in the next Section, as we dive into the new analysis. 

Success guarantees of practical algorithms, such as the Orthogonal Matching Pursuit (OMP) and the Basis Pursuit Denoising (BPDN), have also been derived under this regime. In the same spirit of the aforementioned stability results, the work in \cite{Donoho2006} showed that these methods recover a solution close to the true sparse vector as long as some sparsity constraint, relying on the mutual coherence of the dictionary, is met.

Another useful property for analyzing the success of pursuit methods, initially proposed in \cite{Tropp2004}, is the Exact Recovery Condition (ERC). Formally, one says that the ERC is met for a support $\mathcal{T}$ with a constant $\theta$ whenever
\begin{equation}
	%\text{ERC}(\mathcal{T},\D): \quad
	\theta =  1 - \underset{i \notin \mathcal{T}}{\max} \|\D^{\dagger}_{\mathcal{T}} \d_i \|_1 > 0,
\end{equation}
where we have denoted by $\D_\mathcal{T}^{\dagger}$ the Moore-Penrose pseudoinverse of the dictionary restricted to support $\mathcal{T}$, and $\d_i$ refers to the $i^{th}$ atom in $\D$. Assuming the above is satisfied, the stability of both the OMP and BP was proven in \cite{Tropp2006}. Moreover, in an effort to provide a more intuitive result, the ERC was shown to hold whenever the total number of non-zeros in $\mathcal{T}$ is less than a certain number, which is a function of the mutual coherence and the noise level (and also the value of the smallest non-zero coefficient, in the case of the OMP).

%% ---------------------------------------------------------------------------------------------------------------
%% ---------------------------------------------------------------------------------------------------------------
\subsection{The Convolutional Sparse Model}
\label{Sec:Conv_Model}
We now briefly review the structure of the convolutional sparse model along with the main results from part I of this work. Consider an $N$-dimensional signal $\X=\D\Gama$, where $\D$ is a concatenation of $m$ banded and Circulant matrices, each corresponding to an $n$-dimensional filter in all possible shifts. From another perspective, the $N\times Nm$ dictionary $\D$ can be understood as shifted versions of a local dictionary $\D_L$ of size $n\times m$. Looking at the system of equations corresponding to the $i^{th}$ patch $\x_i$, extracted from the global system through the operator $\R_i$, one can write $\x_i=\R_i \X =\O\gama_i$. The sparse vector $\gama_i$, which is a stripe of length $(2n-1)m$ extracted from $\Gama$, and the corresponding stripe-dictionary $\O$ of size $n\times (2n-1)m$, are both presented in Figure \ref{PartialStripe}, which summarizes this construction. We follow the notation introduced in part I \cite{Papyan2016_1}, and refer the reader to the detailed description therein. As in the preceeding part, we choose to denote global vector with capital letters and local ones with lowercase.

\begin{figure}[t]
	\centering
	\includegraphics[trim = 0 0 50 0 ,width=0.47\textwidth]{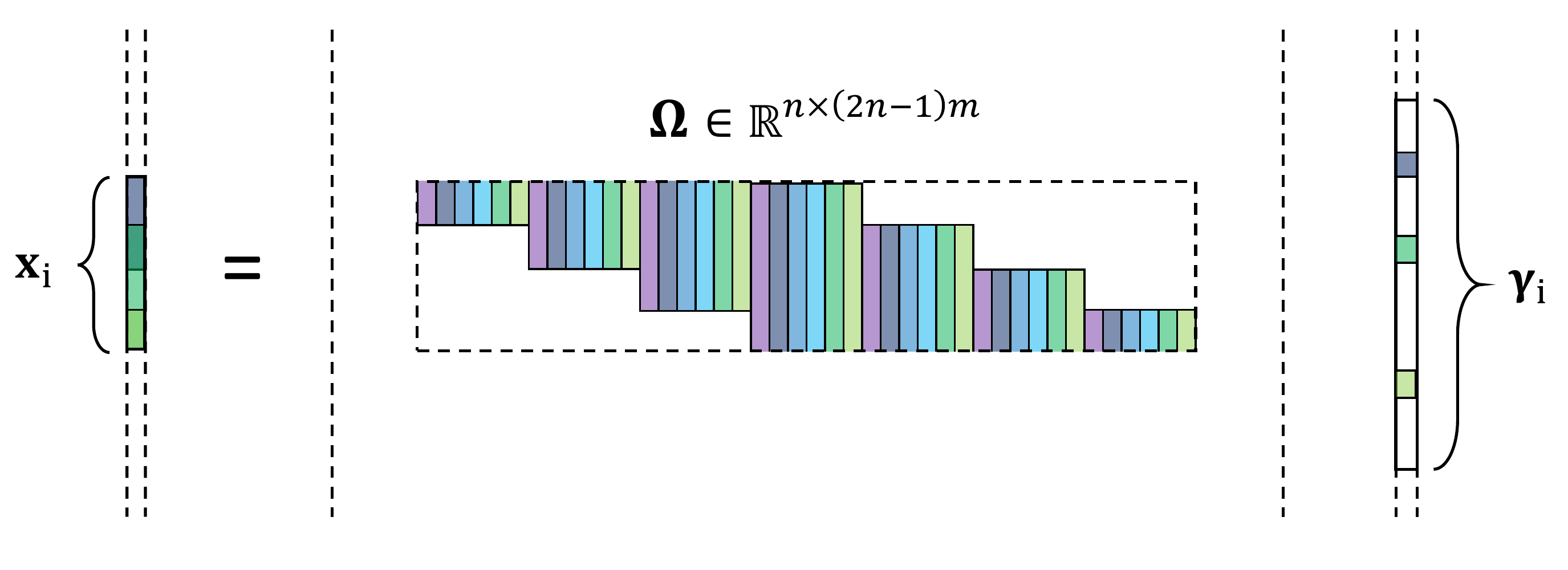}
	\caption{A local stripe from the global system of equations, exhibiting the construction of a patch $\x_i$ in terms of the stripe-dictionary $\O$ and the stripe vector $\gama_i$.} 
	\label{PartialStripe}
	\vspace{-0.3cm}
\end{figure}

In the first part of this work, we have defined the $\Loi$ norm of the sparse vector $\Gama$ to be the maximal $\ell_0$ norm of a stripe $\gama_i$ extracted from it. Formally, this can be written as
\begin{equation}
\|\Gama\|_{0,\infty} = \max_i \|\gama_i\|_0.
\end{equation}
This, in turn, gave rise to the definition of the $\Poi$ problem, where one seeks for the sparsest representation $\Gama$ (in the $\Loi$ sense) of the signal $\X$; i.e.,
\begin{equation}
(\Poi): \quad \min_\Gama \quad \|\Gama\|_{0,\infty} \ \text{ s.t. }\ \D\Gama=\X.
\end{equation}
As it was described in the first part of this work, the shift from the traditional $P_0$ problem to the new $\Poi$ brings about a fundamental advantage in terms of the theoretical guarantees one can provide. In particular, given a solution with a sufficiently small $\Loi$ norm, one can claim the uniqueness of said solution in solving the $\Poi$ problem. Moreover, under the same condition, one is guaranteed to recover this unique minimizer by employing classical pursuit algorithms, such as OMP and BP.

%% ---------------------------------------------------------------------------------------------------------------
%% ---------------------------------------------------------------------------------------------------------------
\section{From Global to Local Stability Analysis } \label{Sec:Global2Local}

Assume a clean signal $\X$, which admits a sparse representation $\Gama$ in terms of the convolutional dictionary $\D$, is contaminated with noise $\E$ (of bounded energy, $\| \E \|_2 \leq \epsilon$)  to create $\Y=\D\Gama+\E$. Given this noisy signal, one could propose to recover the true representation $\Gama$, or a vector close to it, by solving the $P_0^\epsilon$ problem. In this context, as mentioned in the previous section, several theoretical guarantees have been proposed in the literature. As an example, consider the stability results presented in the seminal work of \cite{Donoho2006}. Therein, it was shown that assuming the total number of non-zeros in $\Gama$ is less than $\frac{1}{2}\left(1+\frac{1}{\mu(\D)}\right)$, the distance between the solution to the $P_0^\epsilon$ problem, $\overline{\Gama}$, and the true sparse vector, $\Gama$, satisfies
\begin{equation}
	\|\overline{\Gama} - \Gama\|_2^2 \leq \frac{4\epsilon^2}{1-\mu(\D)(2\|\Gama\|_0-1)}.
	\label{Eq:OriginalStability}
\end{equation}
In the context of our convolutional setting, however, this result provides a weak bound as it constrains the total number of non-zeros to be below a certain threshold. Based on the Welch bound \cite{Welch1974}, it was shown in part I that the maximal number of non-zeros allowed globally in $\Gama$ scales as $O(\sqrt{n})$ -- no matter how large the global dimension $N$ is. This illustrates the futility of the $P_0^\epsilon$ problem in the convolutional framework and the need for an alternative analysis.

Similar to what was done in part I, we begin by re-defining the $P_0^{\epsilon}$ problem into another one, which captures the convolutional structure by relying on the $\Loi$ norm instead. Consider the problem:
\begin{equation}
\quad (\Poie): \quad \underset{\Gama}{\min} \quad \|\Gama\|_{0,\infty} \ \text{ s.t. }\ \|\Y-\D\Gama\|^2_2\leq\epsilon^2.
\end{equation}
In words, given a noisy measurement $\Y$, we seek for the $\Loi$-sparsest representation vector that explains this signal up to an $\epsilon$ error. In what follows, we address the theoretical aspects of this problem and, in particular, study the stability of its solutions and practical yet secured ways for retrieving them.

%% ---------------------------------------------------------------------------------------------------------------
%% ---------------------------------------------------------------------------------------------------------------
\section{Theoretical Analysis}
\label{Sec:TheoreticalAnalysis}

\subsection{Stability of the $\Poie$ Problem}
As expected, one cannot guarantee the uniqueness of the solution to the $\Poie$ problem, as was done for the $\Poi$ in part I. Instead, in this subsection we shall provide a stability claim that guarantees the found solution to be close to the underlying sparse vector that generated $\Y$. In order to provide such an analysis, we commence by arming ourselves with the necessary mathematical tools. 

\begin{defn}
Let $\D$ be a convolutional dictionary. Consider all the sub matrices $\D_\mathcal{T}$, obtained by restricting the dictionary $\D$ to a support $\mathcal{T}$ with an $\Loi$ norm equal to $k$. Define $\delta_k$ as the smallest quantity such that
\begin{equation}
\forall \Delt \quad (1-\delta_k)\|\Delt\|_2^2\leq\|\D_\mathcal{T} \Delt\|_2^2\leq(1+\delta_k)\|\Delt\|_2^2
\end{equation}
holds true for any choice of the support. Then, $\D$ is said to satisfy $k$-SRIP (Stripe-RIP) with constant $\delta_k$.
\end{defn}

Given a matrix $\D$, similar to the Stripe-Spark, computing the SRIP is hard or practically impossible. Thus bounding it using the mutual coherence is of practical use.

\begin{thm}{(Upper bounding the SRIP via the mutual coherence):}
	For a convolutional dictionary $\D$ with global mutual coherence $\mu(\D)$, the SRIP can be upper-bounded by
	\begin{equation}
	\delta_k\leq(k-1)\mu(\D).
	\end{equation}
\end{thm}

\begin{proof}
Consider the sub-dictionary $\D_{\mathcal{T}}$, obtained by restricting the columns of $\D$ to a support $\mathcal{T}$ with $\Loi$ norm equal to $k$. Lemma 1 in part I \cite{Papyan2016_1} states that the eigenvalues of the Gram matrix $\D_{\mathcal{T}}^T\D_{\mathcal{T}}$ are bounded by
\begin{equation} \label{Gersh}
1-(k-1)\mu(\D)\leq\lambda_i(\D_{\mathcal{T}}^T\D_{\mathcal{T}})\leq 1+(k-1)\mu(\D).
\end{equation}
Now, for every $\Delt$ we have that
\begin{align*}
(1-(k-1)\mu(\D))\|\Delt\|_2^2 \leq & \lambda_{min}(\D_{\mathcal{T}}^T\D_{\mathcal{T}})\|\Delt\|_2^2 \\
\leq & \|\D_{\mathcal{T}} \Delt\|_2^2 \leq \lambda_{max}(\D_{\mathcal{T}}^T\D_{\mathcal{T}})\|\Delt\|_2^2 \\
\leq & (1+(k-1)\mu(\D))\|\Delt\|_2^2,
\end{align*}
where $\lambda_{max}$ and $\lambda_{min}$ are the maximal and minimal eigenvalues, respectively.
As a result, we obtain that $\delta_k\leq(k-1)\mu(\D)$.
\end{proof}

Assume a sparse vector $\Gama$ is multiplied by $\D$ and then contaminated by a vector $\E$, generating the signal $\Y=\D\Gama+\E$, such that $\|\Y-\D\Gama\|_2^2\leq\epsilon^2$. Suppose we solve the $\Poie$ problem and obtain a solution $\hat{\Gama}$. How close is this solution to the original $\Gama$? The following theorem provides an answer to this question.

\begin{thm}{(Stability of the solution to the $\Poie$ problem):}
Consider a sparse vector $\Gama$ such that $\|\Gama\|_{0,\infty} = k < \frac{1}{2}\left( 1 + \frac{1}{\mu(\D)} \right) $, and a convolutional dictionary $\D$ satisfying the SRIP property for $\Loi=2k$ with coefficient $\delta_{2k}$. Then, the distance between the true sparse vector $\Gama$ and the solution to the $\Poie$ problem $\hat{\Gama}$ is bounded by
\begin{equation}
\|\Gama-\hat{\Gama}\|_2^2\leq \frac{4\epsilon^2}{1-\delta_{2k}}\leq\frac{4\epsilon^2}{1-(2k-1)\mu(\D)}.	
\label{Eq:NewStability}
\end{equation}
\end{thm}

\begin{proof}
The solution to the $\Poie$ problem satisfies $\|\Y-\D\hat{\Gama}\|_2^2\leq\epsilon^2$, and it must also satisfy $\|\hat{\Gama}\|_{0,\infty}\leq\|\Gama\|_{0,\infty}$ (since $\hat{\Gama}$ is the solution with the minimal $\Loi$ norm).
Defining $\Delt=\Gama-\hat{\Gama}$, using the triangle inequality, we have that $\|\D\Delt\|_2^2=\|\D\Gama-\Y+\Y-\D\hat{\Gama}\|_2^2\leq4\epsilon^2$. Furthermore, since the $\Loi$ norm satisfies the triangle inequality as well, we have that $\|\Delt\|_{0,\infty}=\|\Gama-\hat{\Gama}\|_{0,\infty}\leq\|\Gama\|_{0,\infty}+\|\hat{\Gama}\|_{0,\infty} \leq 2k$.

Using the SRIP of $\D$, we have that
\begin{equation}
(1-\delta_{2k})\|\Delt\|_2^2\leq\|\D\Delt\|_2^2\leq 4\epsilon^2,
\end{equation}
where in the first inequality we have used the lower bound provided by the definition of the SRIP. Finally, we obtain the following stability claim:
\begin{equation}
\|\Delt\|_2^2=\|\Gama-\hat{\Gama}\|_2^2\leq \frac{4\epsilon^2}{1-\delta_{2k}}.
\end{equation}
Using our bound of the SRIP in terms of the mutual coherence, we obtain that
\begin{equation}
\|\Delt\|_2^2=\|\Gama-\hat{\Gama}\|_2^2\leq \frac{4\epsilon^2}{1-\delta_{2k}}\leq\frac{4\epsilon^2}{1-(2k-1)\mu(\D)}.
\end{equation}
For the last inequality to hold, we have assumed $k = \|\Gama\|_{0,\infty}<\frac{1}{2}(1+\frac{1}{\mu(\D)})$.

\end{proof}

One should wonder if the new guarantee presents any advantage when compared to the bound based on the traditional RIP. Looking at the original stability claim for the global system, as discussed in Section \ref{Sec:Global2Local}, 
%we obtain that $\|\Gama-\overline{\Gama}\|_2^2\leq\frac{4\epsilon^2}{1-(2\|\Gama\|_0-1)\mu(\D)}$, with $\overline{\Gama}$ being the solution to the $P_0^\epsilon$ problem (whereas $\hat{\Gama}$ is the solution to the $\Poie$ problem).
%In order to analyze the difference between both bounds, 
the reader should compare the assumptions on the sparse vector $\Gama$, as well as the obtained bounds on the distance between the estimates and the original vector. The stability claim in the $P_0^\epsilon$ problem is valid under the condition
\begin{equation}
\|\Gama\|_0<\frac{1}{2}\left(1+\frac{1}{\mu(\D)}\right).
\end{equation}
In contrast, the stability claim presented above holds whenever
\begin{equation}
\|\Gama\|_{0,\infty}<\frac{1}{2}\left(1+\frac{1}{\mu(\D)}\right).
\end{equation}
This allows for significantly more non-zeros in the global signal, as thoroughly discussed in part I. Furthermore, as long as the above hold and comparing Equations \eqref{Eq:OriginalStability} and \eqref{Eq:NewStability}, we have that
\begin{equation}
\frac{4\epsilon^2}{1-(2\|\Gama\|_{0,\infty}-1)\mu(\D)} \ll \frac{4\epsilon^2}{1-(2\|\Gama\|_0-1)\mu(\D)},
\end{equation}
since generally $\|\Gama\|_{0,\infty}\ll\|\Gama\|_0$. This inequality implies that the above developed bound is (usually much) lower than the traditional one. In other words, the bound on the distance to the true sparse vector is much tighter and far more informative under the $\Loi$ setting.

\vspace{-0.1cm}
\subsection{Stability Guarantee of OMP}
Hitherto, we have shown that the solution to the $\Poie$ problem will be close to the true sparse vector $\Gama$. However, it is also important to know whether this solution can be approximated by pursuit algorithms. In this subsection, we address such a question for the OMP, extending the analysis presented in part I to the noisy setting.

In \cite{Donoho2006}, a claim was provided for the OMP, guaranteeing the recovery of the true support of the underlying solution if
\begin{equation}
	\|\Gama\|_0 < \frac{1}{2}\left(1+\frac{1}{\mu(\D)}\right) - \frac{1}{\mu(\D)}\cdot\frac{\epsilon}{|\Gamma_{min}|},
\end{equation}
$|\Gamma_{min}|$ being the minimal absolute value of the (non-zero) coefficients in $\Gama$. This result comes to show the importance of both the sparsity of $\Gama$ and the signal-to-noise ratio, which relates to the term ${\epsilon}/{|\Gamma_{min}|}$. Nevertheless, in the context of our convolutional setting, this result provides a weak bound for two different reason. First, note that the above bound restricts the total number of non-zeros in the representation of the signal. Following the results from part I, it is natural to seek for an alternative condition for the success of this pursuit relying on the $\Loi$ norm instead. Second, notice that the rightmost term in the above bound divides the global error energy by the minimal coefficient (in absolute value) in $\Gama$. In the convolutional scenario, the energy of the error $\epsilon$ is a \textit{global} quantity, while the minimal coefficient $|\Gamma_{min}|$ is a \textit{local} one -- thus making this term enormous, and the corresponding bound nearly meaningless. As we show next, one can harness the inherent locality of the atoms in order to replace the global quantity in the numerator with a local one: $\epsilon_L$.

\begin{thm}{(Stable recovery of global OMP in the presence of noise):} \label{Theorem:StabilityOMP}
	Suppose a clean signal $\X$ has a representation $\D\Gama$, and that it is contaminated with noise $\E$ to create the signal $\Y=\X+\E$, such that $\|\Y-\X\|_2\leq\epsilon$. Denote by $\epsilon_{_L}$ the highest energy of all $n$-dimensional local patches extracted from $\E$. Assume $\Gama$ satisfies
	\begin{equation} \label{omp_hypothesis}
	\|\Gama\|_{0,\infty} < \frac{1}{2}\left( 1+\frac{1}{\mu(\D)} \right)-\frac{1}{\mu(\D)}\cdot\frac{\epsilon_{_L}}{|\Gamma_{min}|},
	\end{equation}
	where $|\Gamma_{min}|$ is the minimal entry in absolute value of the sparse vector $\Gama$.
	Denoting by $\Gama_\text{OMP}$ the solution obtained by running OMP for $\|\Gama\|_0$ iterations, we are guaranteed that
	\begin{enumerate}[ a) ]
	\item OMP will find the correct support; And,
	\item $\|\Gama_\text{OMP}-\Gama\|_2^2\leq\frac{\epsilon^2}{1-\mu(\|\Gama\|_{0,\infty}-1)}$.
	\end{enumerate}
\end{thm}

The proof of this theorem is presented in Appendix \ref{App:StabilityOMP}, and the derivations therein are based on the analysis presented in \cite{Donoho2006}, generalizing the study to the convolutional setting.

In the theorem presented above, we have assumed that the OMP algorithm runs for $\|\Gama\|_0$ iterations. We could also propose a different approach, however, using a stopping criterion based on the norm of the residual. Under such setting, the OMP would run until the energy of the global residual is less than the energy of the noise, given by $\epsilon^2$. 

\vspace{-0.1cm}
\subsection{Stability Guarantee of Basis Pursuit Denoising via ERC}
Although solving the $\Poi$ problem is at least as hard as solving the $P_0$ version (which is NP-hard), one can nevertheless approximate its solution using the BP algorithm by replacing the $\Loi$ norm with the convex $\ell_1$. A different and perhaps more appropriate approach could be suggested, relying on the $\ell_{1,\infty}$ norm. This, however, remains one of our future work challenges. A theoretical motivation behind the $\ell_1$ relaxation was proven in part I, showing that assuming the $\Loi$ norm of the underlying solution is low, the BP algorithm is guaranteed to find it. When moving to the noisy regime, the BP is naturally extended to the Basis Pursuit DeNoising (BPDN) algorithm\footnote{Note that an alternative to the BPDN extension is that of the Dantzig Selector algorithm. One can envision a similar analysis to the one presented here for this algorithm as well.}, which in its Lagrangian form is defined as follows
\begin{equation} \label{eq:lagrangian_BP}
\min_{\Gama} \frac{1}{2}\|\Y - \D \Gama \|^2_2 +\lambda \| \Gama \|_1.
\end{equation}
Similar to the way part I has shown how BP can be used to approximate the $\Poi$ problem, in what follows we will prove that the BPDN manages to approximate the solution to the $\Poie$ problem.

Assuming the ERC is met, the stability of BP was proven under various noise models and formulations in \cite{Tropp2006}. By exploiting the convolutional structure used throughout our analysis, we now show that the ERC is met given that the $\Loi$ norm is small, tying the aforementioned results to our story.

\begin{thm}{(ERC in the convolutional sparse model):}
\label{Theorem:ERC_Loi}
	For a convolutional dictionary $\D$ with mutual coherence $\mu(\D)$, the ERC condition is met for every support $\mathcal{T}$ that satisfies\footnote{Note that specifying the $\Loi$ of a support rather than a sparse vector is a slight abuse of notation, that we will nevertheless use for the sake of simplicity, as was done in part I.}
	\begin{equation}
		\|\mathcal{T}\|_{0,\infty} < \frac{1}{2}\left(1+\frac{1}{\mu(\D)}\right).
	\end{equation}
\end{thm}

Based on this and the analysis presented in \cite{Tropp2006}, we present a stability claim for the Lagrangian formulation of the BP problem as stated in Equation \eqref{eq:lagrangian_BP}.
\begin{thm}{(Stable recovery of global Basis Pursuit in the presence of noise):} \label{Theorem:StabilityBP}
	Suppose a clean signal $\X$ has a representation $\D\Gama$, and that it is contaminated with noise $\E$ to create the signal $\Y=\X+\E$. Denote by $\epsilon_{_L}$ the highest energy of all $n$-dimensional local patches extracted from $\E$. Assume $\Gama$ satisfies
	\begin{equation} \label{eq:BP_assumption}
	\|\Gama\|_{0,\infty}<\frac{1}{3} \left( 1 + \frac{1}{\mu(\D)} \right).
	\end{equation}
	Denoting by $\Gama_{\text{BP}}$ the solution to the Lagrangian BP formulation with parameter $\lambda=4\epsilon_L$, we are guaranteed that
	\begin{enumerate}
	\item The support of $\Gama_{\text{BP}}$ is contained in that of $\Gama$.
	\item $\|\Gama_{\text{BP}}-\Gama\|_\infty<\frac{15}{2}\epsilon_L$.
	\item In particular, the support of $\Gama_{\text{BP}}$ contains every index $i$ for which $|\Gamma_i|>\frac{15}{2}\epsilon_L$.
	\item The minimizer of the problem, $\Gama_{\text{BP}}$, is unique.
	\end{enumerate}
\end{thm}
\noindent
The proof for both of the above, inspired by the derivations in \cite{Elad_Book} and \cite{Tropp2006}, are presented in Appendix \ref{App:ERC_Loi} and \ref{App:StabilityBP}.

The benefit of this over traditional claims is, once again, the replacement of the $\ell_0$ with the $\Loi$ norm. Moreover, this result bounds the difference between the entries in $\Gama_{\text{BP}}$ and $\Gama$ in terms of a local quantity -- the local noise level $\epsilon_L$. As a consequence, all atoms with coefficients above this local measure are guaranteed to be recovered.

The implications of the above theorem are far-reaching as it provides a sound theoretical back-bone for all works that have addressed the convolutional BP problem in its Lagrangian form \cite{Bristow2013,Wohlberg2016,Bristow2014,Heide2015,Kong2014}. Later, in Section \ref{Sec:FromGlobal2LocalProcessing}, we shall propose two algorithms for solving the global BP efficiently by working locally, and these methods would benefit from this theoretical result as well.

\subsection{Experiments}

Following the above analysis, we now provide a numerical experiment demonstrating the above obtained bounds. The global dictionary employed here is the same as the one used for the noiseless experiments in part I, with mutual coherence $\mu(\D)=0.09$, local atoms of length $n=64$ and global ones of size $N = 640$. We sample random sparse vectors with cardinality between $1$ and $500$, with entries drawn from a uniform distribution with range $\left[-a,a\right]$, for varying values of $a$. Given these vectors, we construct global signals and contaminate them with noise. The noise is sampled from a zero-mean unit-variance white Gaussian distribution, and then normalized such that $\| \E \|_2 = 0.1$.

%\begin{figure*}[t]
%	\centering
%	\includegraphics[trim = 50 0 50 0, width = .85\textwidth]{PhaseTransition_OMP_Noisy.pdf}
%	\caption{The ratio $\epsilon_{_L}/|\Gama_{\min}|$ as a function of the $\Loi$ norm, and the theoretical bound for the successful recovery of the support, for the OMP algorithm.}
%	\label{fig:PhaseTransition_OMP_Noisy}
%\end{figure*}
%\begin{figure*}[t]
%	\centering
%	\includegraphics[trim = 50 0 50 0, width = .85\textwidth]{PhaseTransition_BP_Noisy.pdf}
%	\caption{The ratio $\epsilon_{_L}/|\Gama_{\min}|$ as a function of the $\Loi$ norm, and the theoretical bound for the successful recovery of the support, for the BP algorithm.}
%	\label{fig:PhaseTransition_BP_Noisy}
%\end{figure*}
\begin{figure}[t]
\centering
\includegraphics[trim = 50 0 30 20, width = 0.35\textwidth]{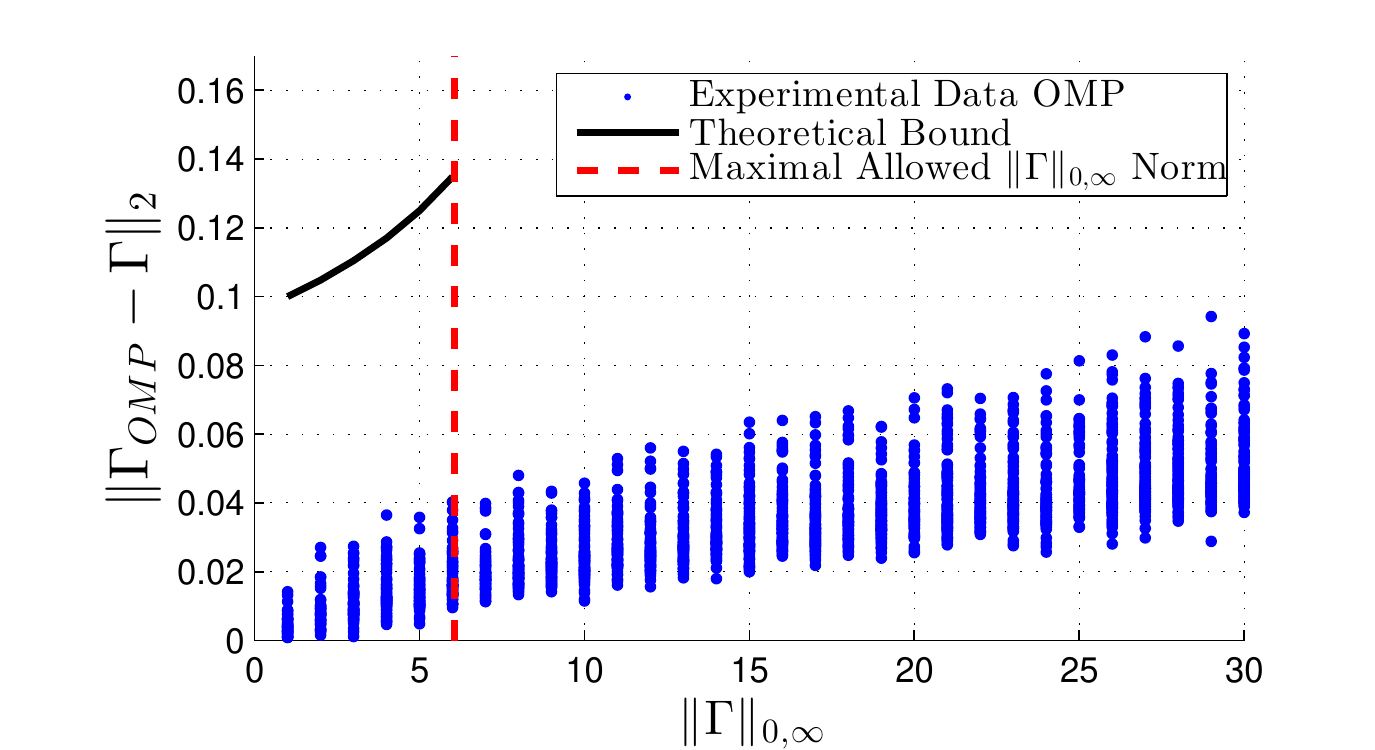}
\caption{The distance $\|\Gama_{\text{OMP}} - \Gama\|_2$ as a function of the $\Loi$ norm, and the corresponding theoretical bound.}
\label{fig:L2Dist_OMP_Noisy}
\vspace{-0.2cm}
\end{figure}

In what follows, we will first center our attention on the bounds obtained for the OMP algorithm, and then proceed to the ones corresponding to the BP. Given the noisy signals, we run OMP with a sparsity constraint, obtaining $\Gama_{\text{OMP}}$. For each realization of the global signal, we compute the minimal entry (in absolute value) of the global sparse vector, $|\Gamma_{min}|$, and its $\Loi$ norm. In addition, we compute the maximal local energy of the noise, $\epsilon_L$, corresponding to the highest energy of a $n$-dimensional patch of $\E$.

Recall that the theorem in the previous subsection poses two claims: 1) the stability of the result in terms of $\|\Gama_{\text{OMP}} - \Gama\|_2$; and 2) the success in recovering the correct support. In Figure \ref{fig:L2Dist_OMP_Noisy} we investigate the first of these points, presenting the distance between the estimated and the true sparse codes as a function of the $\Loi$ norm of the original vector. As it is clear from the graph, the empirical distances are below the theoretical bound depicted in black, given by $\frac{\epsilon^2}{1-\mu(\D)(\|\Gama\|_{0,\infty}-1)}$. According to the theorem's assumption, the sparse vector should satisfy $\|\Gama\|_{0,\infty} < \frac{1}{2}\left(1 + \frac{1}{\mu(\D)} \right) - \frac{1}{\mu(\D)}\cdot\frac{\epsilon_{_L}}{|\Gamma_{\min}|}$. The red dashed line delimits the area where this is met, with the exception that we omit the second term in the previous expression, as done previously in \cite{Donoho2006}. This disregards the condition on the $|\Gamma_{\min}|$ and $\epsilon_{_L}$ (which depends on the realization). Yet, the empirical results remain stable. 

In order to address the successful recovery of the support, we compute the ratio $\frac{\epsilon_{_L}}{|\Gamma_{\min}|}$ for each realization in the experiment. In Figure \ref{fig:PhaseTransition_OMP_Noisy}, for each sample we denote by $\bullet$ or $\times$ the success or failure in recovering the support, respectively. Each point is plotted as a function of its $\Loi$ norm and its corresponding ratio. The theoretical condition for the success of the OMP can be rewritten as $\frac{\epsilon_{_L}}{|\Gamma_{min}|} < \frac{\mu(\D)}{2}\left( 1+\frac{1}{\mu(\D)} \right) - \mu(\D)\|\Gama\|_{0,\infty}$, presenting a bound on the ratio $\frac{\epsilon_{_L}}{|\Gamma_{\min}|}$ as a function of the $\Loi$ norm. This bound is depicted with a blue line, indicating that the empirical results agree with the theoretical claims.

\begin{figure}[t]
	\centering
	\begin{subfigure}[t]{.5\textwidth}
		\centering
		\includegraphics[trim = 15 0 15 0,width = 1\textwidth]{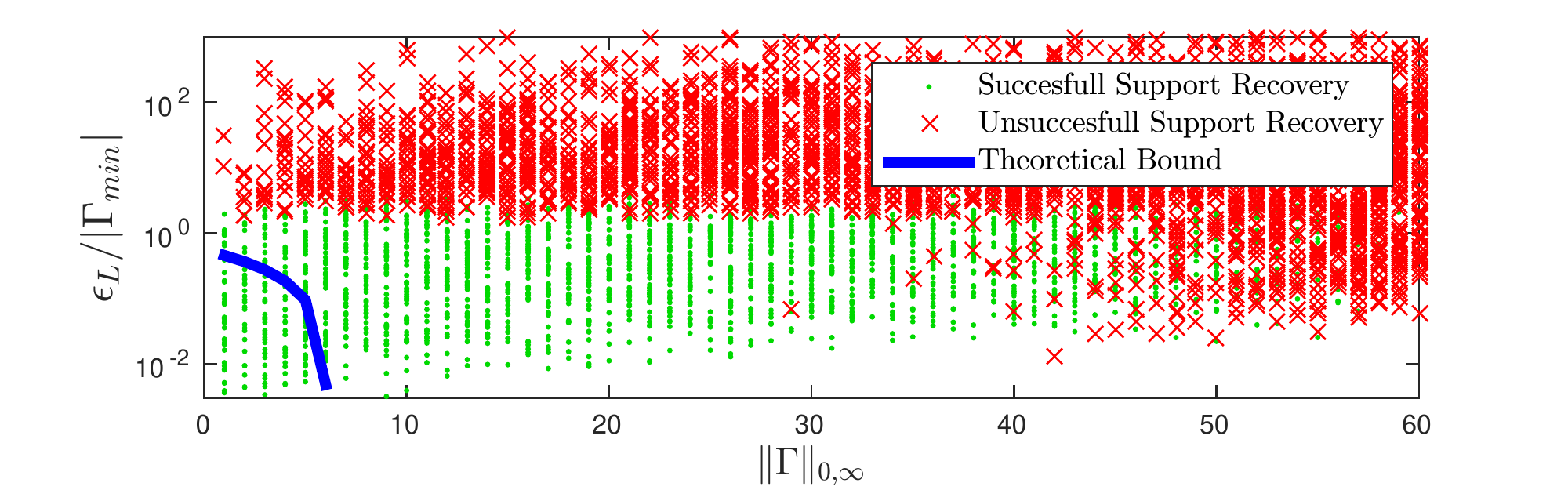}
		\caption{Orthogonal Matching Pursuit.}
		\label{fig:PhaseTransition_OMP_Noisy}
	\end{subfigure} 
	\\[.15cm]
	\begin{subfigure}[t]{.5\textwidth}
		\centering
		\includegraphics[trim = 15 0 15 0,width = 1\textwidth]{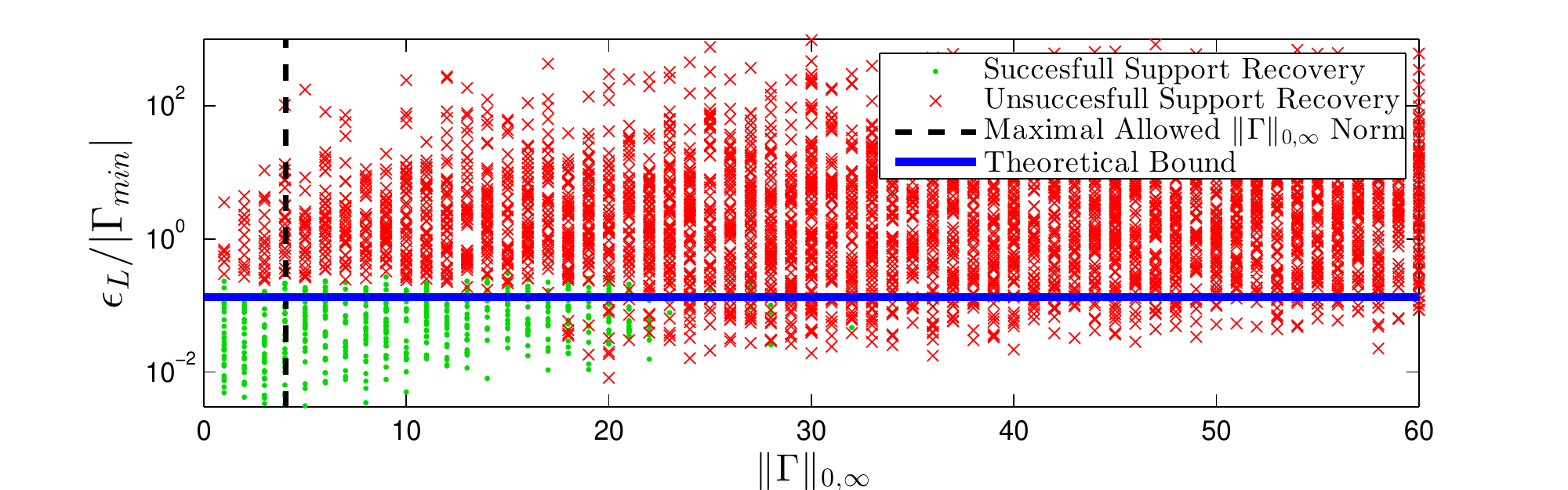}
		\caption{Basis Pursuit.}
		\label{fig:PhaseTransition_BP_Noisy}
	\end{subfigure}
	\caption{The ratio $\epsilon_{_L}/|\Gamma_{\min}|$ as a function of the $\Loi$ norm, and the theoretical bound for the successful recovery of the support, for both the OMP and BP algorithms.}
	\vspace{-0.3cm}
\end{figure}

One can also observe two distinct phase transitions in Figure \ref{fig:PhaseTransition_OMP_Noisy}. On the one hand, noting that the $y$ axis can be interpreted as the inverse of the noise-to-signal ratio (in some sense), we see that once the noise level is too high, OMP fails in recovering the support\footnote{Note that the abrupt change in this phase-transition area is due to the log scale of the $y$ axis.}. On the other hand, similar to what was presented in the noiseless case, once the $\Loi$ norm becomes too large, the algorithm is prone to fail in recovering the support.

We now shift to the empirical verification of the guarantees obtained for the BP. We employ the same dictionary as in the experiment above, and the signals are constructed in the same manner. We use the implementation of the LARS algorithm within the SPAMS package\footnote{Freely available from http://spams-devel.gforge.inria.fr/.} in its Lagrangian formulation with the theoretically justified parameter $\lambda=4\epsilon_L$, obtaining $\Gama_{\text{BP}}$. Once again, we compute the quantities: $|\Gamma_{min}|$, $\|\Gama\|_{0,\infty}$ and $\epsilon_L$.

Theorem \ref{Theorem:StabilityBP} states that the $\ell_\infty$ distance between the BP solution and the true sparse vector is below $\frac{15}{2}\epsilon_L$. In Figure \ref{fig:LInfDist_BP_Noisy} we depict the ratio $\frac{\|\Gama_{\text{BP}}-\Gama\|_\infty}{\epsilon_L}$ for each realization, verifying it is indeed below $\frac{15}{2}$ as long as the $\Loi$ norm is below $\frac{1}{3}\left(1+\frac{1}{\mu(\D)}\right) \approx 4$. Next, we would like to corroborate the assertions regarding the recovery of the true support. To this end, note that the theorem guarantees that all entries satisfying $|\Gamma_i|>\frac{15}{2}\epsilon_L$ shall be recovered by the BP algorithm. Alternatively, one can state that the complete support must be recovered as long as $\frac{\epsilon_L}{|\Gamma_{\min}|}<\frac{2}{15}$. To verify this claim, we plot this ratio for each realization as function of the $\Loi$ norm in Figure \ref{fig:PhaseTransition_BP_Noisy}, marking every point according to the success or failure of BP (in recovering the complete support).
As evidenced in \cite{Elad_Book}, OMP seems to be far more accurate than the BP in recovering the true support. As one can see by comparing Figure \ref{fig:PhaseTransition_OMP_Noisy} and \ref{fig:PhaseTransition_BP_Noisy}, BP fails once the $\Loi$ norm goes beyond $20$, while OMP succeeds all the way until $\|\Gama\|_{0,\infty}=40$. 

%% ---------------------------------------------------------------------------------------------------------------
%% ---------------------------------------------------------------------------------------------------------------
\section{From Global Pursuit to Local Processing}
\label{Sec:FromGlobal2LocalProcessing}

We now turn to analyze the practical aspects of solving the $\Poie$ problem given the relationship $\Y = \D\Gama + \E$. Motivated by the theoretical guarantees of success derived in the previous section, the first na\"ive approach would be to employ global pursuit methods such as OMP and BP. However, these are computationally demanding as the dimensions of the convolutional dictionary are prohibitive for high values of $N$, the signal length.  

As an alternative, one could attempt to solve the $\Poie$ problem using a patch-based processing scheme. In this case, for example, one could suggest to solve a local and relatively cheaper pursuit for every patch in the signal (including overlaps) using the local dictionary $\D_L$. It is clear, however, that this approach will not work well under the convolutional model, because atoms used in overlapping patches are simply not present in $\D_L$. On the other hand, one could turn to employ $\O$ as the \emph{local} dictionary, but this is prone to fail in recovering the correct support of the atoms. To see this more clearly, note that there is no way to distinguish between any of the atoms having only one entry different than zero; i.e., those appearing on the extremes of $\O$ in Figure \ref{PartialStripe}.

As we can see, neither the na\"ive global approach, nor the simple patch-based processing, provide an effective strategy. Several questions arise from this discussion: Can we solve the global pursuit problem using local patch-based processing? Can the proposed algorithm rely merely on the low dimensional dictionaries $\D_L$ or $\O$ while still fully solving the global problem? If so, in what form should the local patches communicate in order to achieve a global consensus? In what follows, we address these issues and provide practical and globally optimal answers.

\begin{figure}[t]
	\centering
	\includegraphics[trim = 50 0 30 20, width = 0.35\textwidth]{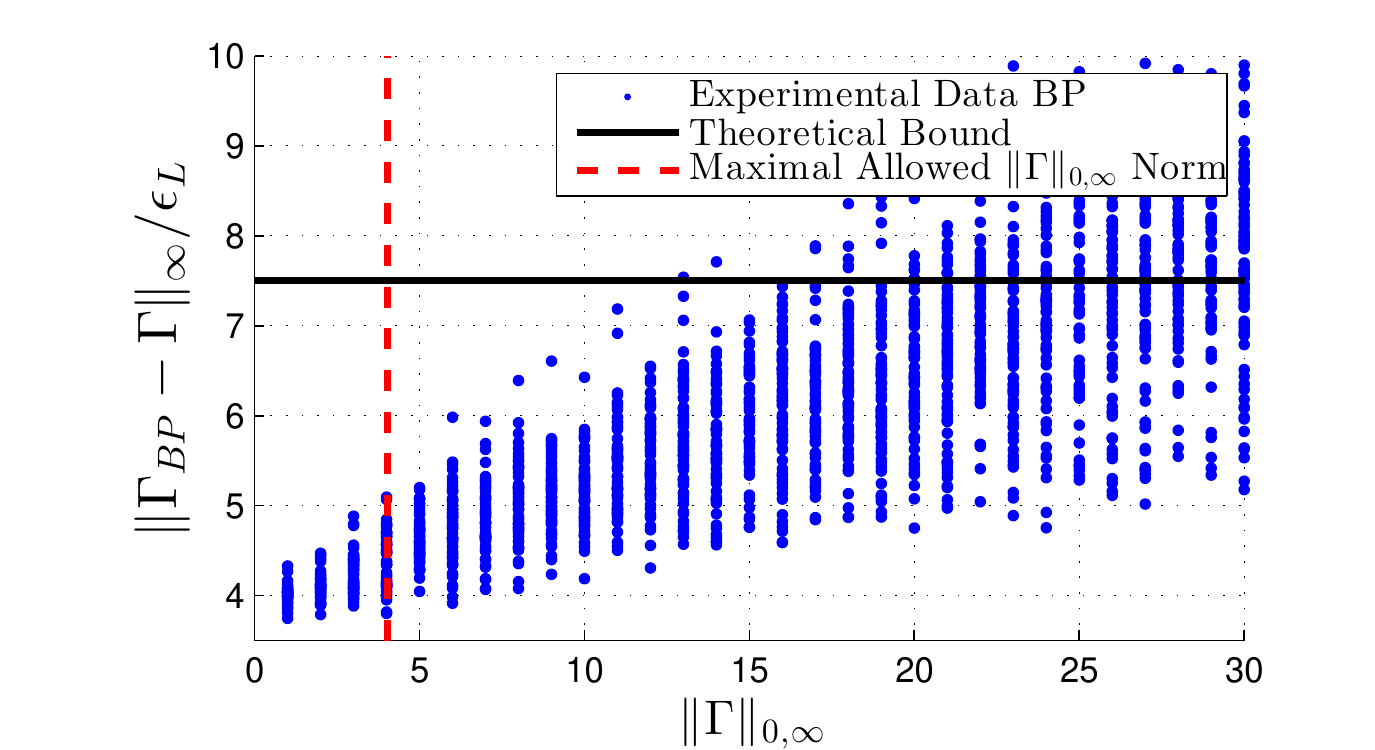}
	\caption{The distance $\|\Gama_{\text{BP}} - \Gama\|_\infty/ \epsilon_L$ as a function of the $\Loi$ norm, and the corresponding theoretical bound.}
	\label{fig:LInfDist_BP_Noisy}
	\vspace{-0.3cm}
\end{figure}

\vspace{-0.1cm}
\subsection{Global to Local Through Bi-Level Consensus}

When dealing with global problems which can be solved locally, a popular tool of choice is the Alternating Direction Method of Multipliers (ADMM) \cite{Boyd2011} in its consensus formulation. In this framework, a global objective can be decomposed into a set of local and distributed problems which attempt to reach a global agreement. We will show that this scheme can be effectively applied in the convolutional sparse coding context, providing an algorithm with a bi-level consensus interpretation.

The ADMM has been extensively used throughout the literature in convolutional sparse coding. However, as explained in the introduction, it has been applied usually in the Fourier domain. As a result, the sense of locality is lost in these approaches and the connection to traditional (local) sparse coding is non-existent. On the contrary, the pursuit method we propose here is carried out in a localized fashion in the original domain, while still benefiting from the advantages of ADMM.

Recall the $\ell_1$ relaxation of the global pursuit, given by
\begin{equation} \label{eq:BP_global}
\min_{\Gama} \frac{1}{2}\|\Y - \D \Gama \|^2_2 +\lambda \| \Gama \|_1.
\end{equation}
Note that the noiseless model is contained in this formulation as a particular case when $\lambda$ tends to zero.
Using the separability of the $\ell_1$ norm, $\| \Gama \|_1 = \sum_i \|\alfa_i\|_1$, where $\alfa_i$ are $m-$dimensional local sparse vectors, as previously defined in part I. In addition, using the fact that $\R_i\D\Gama = \O \gama_i$, we apply a local decomposition on the first term as well. This results in 
\begin{equation}
\min_{\{\alfa_i\},\{\gama_i\}} \quad \frac{1}{2n} \sum_i \| \R_i \Y - \O \gama_i \|^2_2 +\lambda \sum_i \|\alfa_i\|_1,
\end{equation}
where we have divided the sum in the first term by the number of contributions per entry in the global signal, which is equal to the patch size $n$.
Note that the above minimization is not equivalent to the original problem in Equation \eqref{eq:BP_global} since no consensus is enforced between the local variables. Recall that the different $\gama_i$ overlap, and as such the above minimization must enforce them to agree. In addition, $\alfa_i$ should be constrained to be equal to the center of the corresponding $\gama_i$. Based on these observations, we modify the above problem by adding the appropriate constraints, obtaining
\begin{align}
\min_{\{\alfa_i\},\{\gama_i\},\Gama} \quad \frac{1}{2n} \sum_i \| \R_i \Y  &- \O \gama_i \|^2_2 +\lambda \sum_i \|\alfa_i\|_1 \\
& \text{ s.t. } \begin{cases} \mathbf{Q} \gama_i = \alfa_i \\ \S_i \Gama = \gama_i \end{cases} \forall i,
\end{align}
where $\mathbf{Q}$ extracts the center $m$ coefficients corresponding to $\alfa_i$ from $\gama_i$, and $\S_i$ extracts the $i^{th}$ stripe $\gama_i$ from $\Gama$.

Defining $f_i(\gama_i) = \frac{1}{2n} \| \R_i \Y-  \O \gama_i \|^2_2 $ and $g(\alfa_i) = \lambda \| \alfa_i \|_1 $, we can now express our problem as follows
\begin{equation}
\min_{\{\alfa_i\},\{\gama_i\},\Gama}  \sum_i f_i(\gama_i) +  g(\alfa_i)\ \text{ s.t. }  \begin{cases} \mathbf{Q}\gama_i = \alfa_i \\ \S_i \Gama = \gama_i \end{cases} \forall i.
\end{equation}
This is a two-level local-global consensus formulation:  each $m$ dimensional vector $\alfa_i$ is enforced to agree with the center of its corresponding $(2n-1)m$ dimensional $\gama_i$, and in addition, all $\gama_i$ are required to agree with each other as to create a global $\Gama$. 
The above can be shown to be equivalent to the standard two-block ADMM formulation \cite{Boyd2011}. 

Writing the augmented Lagrangian (in its scaled form), we obtain the following objective for the problem above
\begin{align*}
\min_{\Gama,\{\alfa_i\},\{\gama_i\},\{\u_i\},\{\uh_i\}} \sum_i f_i(\gama_i) + g(\alfa_i) & + \frac{\rho}{2} \|  \mathbf{Q}\gama_i - \alfa_i + \u_i  \|^2_2 \\
& + \frac{\rho}{2} \| \S_i \Gama - \gama_i  + \uh_i  \|^2_2,
\end{align*}
which can be minimized with the method depicted in Algorithm \ref{alg:ADMM_algo}. We have introduced the (scaled) Lagrange multipliers $\u_i$ and $\uh_i$ corresponding to the variables $\alfa_i$ and $\gama_i$, respectively, and have denoted by $\rho$ the step size in the algorithm. Each iteration of this method can be divided into four steps:
\begin{enumerate}
\item Local sparse coding that updates $\alfa_i$ (for all $i$), which amounts to a simple soft thresholding operation.
\item Solution of a linear system of equations for updating $\gama_i$ (for all $i$), which boils down to a simple multiplication by a constant matrix.
\item Update of the global sparse vector $\Gama$, which aggregates the $\gama_i$ by averaging.
\item Update of the dual variables.
\end{enumerate}
As can be seen, the ADMM provides a simple way of breaking the global pursuit into local operations. Moreover, the local coding step is just a projection problem onto the $\ell_1$ ball, which can be solved through simple soft thresholding, implying that there is no complex pursuit involved.

\begin{algorithm}[t]
 \While{not converged}{
 
 \vspace{0.3cm}
 Local Thresholding: $\alfa_i \leftarrow \underset{\alfa}{\min}\ \lambda \| \alfa \|_1 + \frac{\rho}{2} \|  \mathbf{Q}\gama_i - \alfa + \u_i  \|^2_2$ \;
  
  \vspace{0.3cm}
  Stripe Projection:
  \begin{align}
  \hspace{-0.7321cm} \gama_i \leftarrow \M^{-1} \left( \frac{1}{n}\O^T\R_i\Y \right. & + {\rho}( \S_i\Gama + \uh_i) \\
  & + \rho \mathbf{Q}^T ( \alfa_i-\u_i ) \Big),
  \end{align}
  where $\M = \rho \mathbf{Q}^T\mathbf{Q} + \frac{1}{n}\O^T\O + {\rho} \mathbf{I}$\;
  
  \vspace{0.3cm}
  Global Update:\newline
  $\Gama \leftarrow \left( \sum_i \S_i^T \S_i \right)^{-1} \sum_i \S_i^T (\gama_i - \uh_i) $ \;
 
  \vspace{0.3cm}
  Dual Variables Update:\newline
  $\u_i \leftarrow \u_i + (\mathbf{Q} \gama_i - \alfa_i) $ \; 
  $\uh_i \leftarrow \uh_i + (\S_i \Gama - \gama_i) $ \;
 }
 \caption{Global pursuit using local processing via ADMM.}
\label{alg:ADMM_algo}
\end{algorithm}

Since we are in the $\ell_1$ case, the function $g$ is convex, and so are the functions $f_i$. Therefore, the above is guaranteed to converge to the minimizer of the global BP problem. As a result, we benefit from the theoretical guarantees derived in previous sections. One could attempt, in addition, to enforce an $\ell_0$ penalty instead of the $\ell_1$ norm on the global sparse vector. Despite the fact that no convergence guarantees could be claimed under such formulation, the derivation of the algorithm remains practically the same, with the only exception that the soft thresholding is replaced by a hard one.

\vspace{-0.3cm}
\subsection{An Iterative Soft Thresholding Approach}
While the above algorithm suggests a way to tackle the global problem in a local fashion, the matrix involved in the stripe projection stage is relatively large when compared to the dimensions of $\D_L$. As a consequence, the bi-level consensus introduces an extra layer of complexity to the algorithm. In what follows, we propose an alternative method based on the Iterative Soft Thresholding (IST) algorithm which relies solely on multiplications by $\D_L$, which features a simple intuitive interpretation and implementation. A similar approach for solving the convolutional sparse coding problem was suggested in \cite{Chalasani2013}. Our main concern here is to provide insight into local alternatives for the global sparse coding problem and their guarantees, whereas the work in \cite{Chalasani2013} focused on the optimizations aspects of this pursuit from an entirely global perspective.

\begin{algorithm}[t]
	$\forall i \quad \r_i^0 = \R_i \Y, \quad \alfa_i^0 = \mathbf{0}$\;
	k = 1\;
	\While{not converged}{
		\vspace{0.3cm}
		Local Coding:\newline $\forall i \quad \alfa_i^k = \mathcal{S}_{\lambda/c}\left( \alfa_i^{k-1} + \frac{1}{c} \ \D_L^T \ \r_i^{k-1} \right)$ \;
		
		\vspace{0.3cm}
		Computation of the Patch Averaging Aggregation:\newline
		$ \widehat{\X}^k = \sum_i \R_i^T\D_L\alfa_i^k $ \; 
				
		\vspace{0.3cm}
		Update of the Residuals:\newline
		$\forall i \quad \r_i^k = \R_i \left( \Y - \widehat{\X}^k \right)$ \; 
		
		\vspace{0.3cm}
		$k=k+1$\;
	}
	\caption{Global pursuit using local processing via iterative soft thresholding.}
	\label{alg:IST_algo}
\end{algorithm}

Let us consider the IST algorithm \cite{Daubechies2004} which minimizes the global objective
\begin{equation}
\min_{\Gama} \frac{1}{2}\|\Y - \D \Gama \|^2_2 +\lambda \| \Gama \|_1
\end{equation}
by iterating the following updates
\begin{equation}
\Gama^k = \mathcal{S}_{\lambda/c}\left( \Gama^{k-1} + \frac{1}{c} \D^T(\Y-\D\Gama^{k-1}) \right).
\end{equation}
Given a vector, the operator $\mathcal{S}$ applies a soft thresholding with threshold $\lambda/c$ on its entries. Interpreting the above as a projected gradient descent, the coefficient $c$ relates to the gradient step size and should be set according to the maximal singular value of the matrix $\D$ in order to guarantee convergence \cite{Daubechies2004}.

The above algorithm might at first seem undesirable due to the multiplications of the residual $\Y-\D\Gama^{k-1}$ with the global dictionary $\D$. Yet, we will show that such a multiplication does not need to be carried out explicitly due to the convolutional structure imposed on our dictionary.

\begin{figure}[!htb]
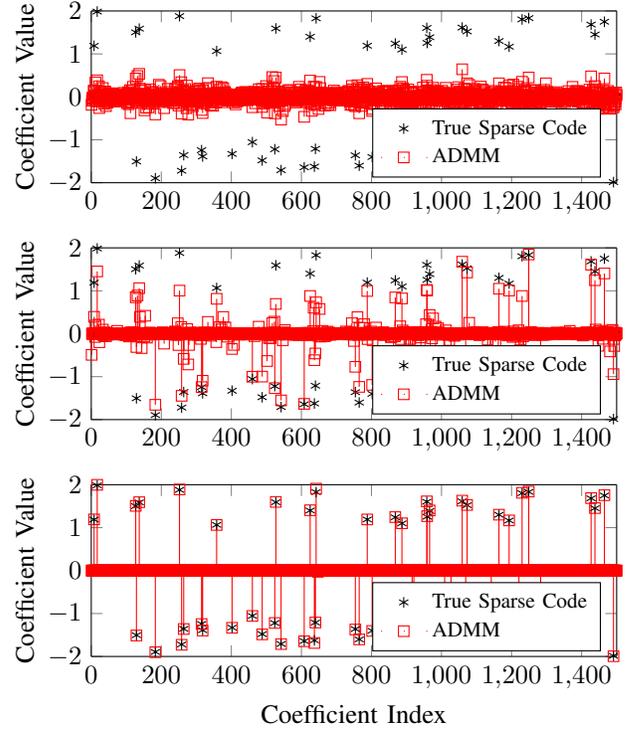

	\centering
	\begin{subfigure}{.45\textwidth}
		%\caption{Iteration number 20.}
		\setlength\figureheight{.28\textwidth}
		\setlength\figurewidth{.9\textwidth}
		\input{ADMM_L1_noiseless_codes_1.tikz}
	\end{subfigure}
	
	\centering
	\vspace{0.1cm}
	\begin{subfigure}{.45\textwidth}
		%\caption{Iteration number 200.}
		\setlength\figureheight{.28\textwidth}
		\setlength\figurewidth{.9\textwidth}
		\input{ADMM_L1_noiseless_codes_2.tikz}
	\end{subfigure}
	
	\centering
	\vspace{0.1cm}
	\begin{subfigure}{.45\textwidth}
		%\caption{Iteration number 1000.}
		\setlength\figureheight{.28\textwidth}
		\setlength\figurewidth{.9\textwidth}
		\input{ADMM_L1_noiseless_codes_3.tikz}
	\end{subfigure}
	\caption{The sparse vector $\Gama$ after the global update stage in the ADMM algorithm at iterations $20$ (top), $200$ (middle) and $1000$ (bottom). An $\ell_1$ norm formulation was used for this experiment, in a noiseless setting.}
	\label{fig:ADMM_L1_noiseless_codes}
	\vspace{-0.4cm}
\end{figure}

Defining as $\P_i$ the operator which extracts the $i^{th}$ $m-$dimensional vector from $\Gama$, we can break the above algorithm into local updates by
\begin{equation}
\P_i\Gama^k = \mathcal{S}_{\lambda/c}\left( \P_i\Gama^{k-1} + \frac{1}{c} \P_i\D^T(\Y-\D\Gama^{k-1}) \right).
\end{equation}
As a first observation, the matrix $\P_i\D^T$, which is of size $m \times N$, is in-fact $\D_L^T$ padded with zeros. As a consequence, the above can be rewritten as follows:
\begin{equation}
\P_i\Gama^k = \mathcal{S}_{\lambda/c}\left( \P_i\Gama^{k-1} + \frac{1}{c} \P_i\D^T\R_i^T \R_i (\Y-\D\Gama^{k-1}) \right),
\end{equation}
where we have used $\R_i$ as the operator which extracts the $i^{th}$ $n-$dimensional patch from an $N-$dimensional global signal. The operator $\P_i$ extracts $m$ rows from $\D^T$, while $\R_i^T$ extracts its non-zero columns. Therefore, $\P_i\D^T\R_i^T = \D_L^T$, and so we can write
\begin{equation}
\P_i\Gama^k = \mathcal{S}_{\lambda/c}\left( \P_i\Gama^{k-1} + \frac{1}{c} \D_L^T \R_i (\Y-\D\Gama^{k-1}) \right).
\end{equation}
Noting that $\alfa_i^k=\P_i\Gama^k$ is the $i^{th}$ local sparse code, and defining $\r_i^k=\R_i (\Y-\D\Gama^{k-1})$ as the corresponding patch-residual at iteration $k$, we obtain our final update (for every patch)
\begin{equation}
\alfa_i^k = \mathcal{S}_{\lambda/c}\left( \alfa_i^{k-1} + \frac{1}{c} \ \D_L^T \ \r_i^{k-1} \right).
\end{equation}
We summarize the above derivations in Algorithm \ref{alg:IST_algo}.

As we see, all operations can be expressed in terms of low dimensional $\alfa_i$ and the small dictionary $\D_L$. Moreover, we can interpret each iteration of this algorithm as a scatter and gather process. Given a global signal, local residuals are first extracted and scattered to different nodes, where they undergo local shrinkage operations. Then, their results are gathered for the re-computation of the global residual.

Assuming the step size is chosen appropriately, as explained previously, the above algorithm is guaranteed to converge to the solution of the global BP. As such, our theoretical analysis holds in this case as well. Alternatively, one could attempt to employ an $\ell_0$ approach, using a global iterative hard thresholding algorithm. In this case, however, there is no theoretical guarantees in terms of the $\Loi$ norm. Still, we believe that a similar analysis to the one taken throughout this work could lead to such claims.

\vspace{-0.3cm}
\subsection{Experiments}
Next, we proceed to provide empirical results for the above described methods. To this end, we take an undercomplete DCT dictionary of size $25\times 5$, and use it as $\D_L$ in order to construct the global convolutional dictionary $\D$ for a signal of length $N = 300$. We then generate a random global sparse vector $\Gama$ with $50$ non-zeros, with entries distributed uniformally in the range $[-2,-1]\ \cup\ [1,2]$, creating the signal $\X = \D\Gama$.

\begin{figure}[t]
	\centering
	\includegraphics[trim = 30 10 10 10, width = 0.48\textwidth]{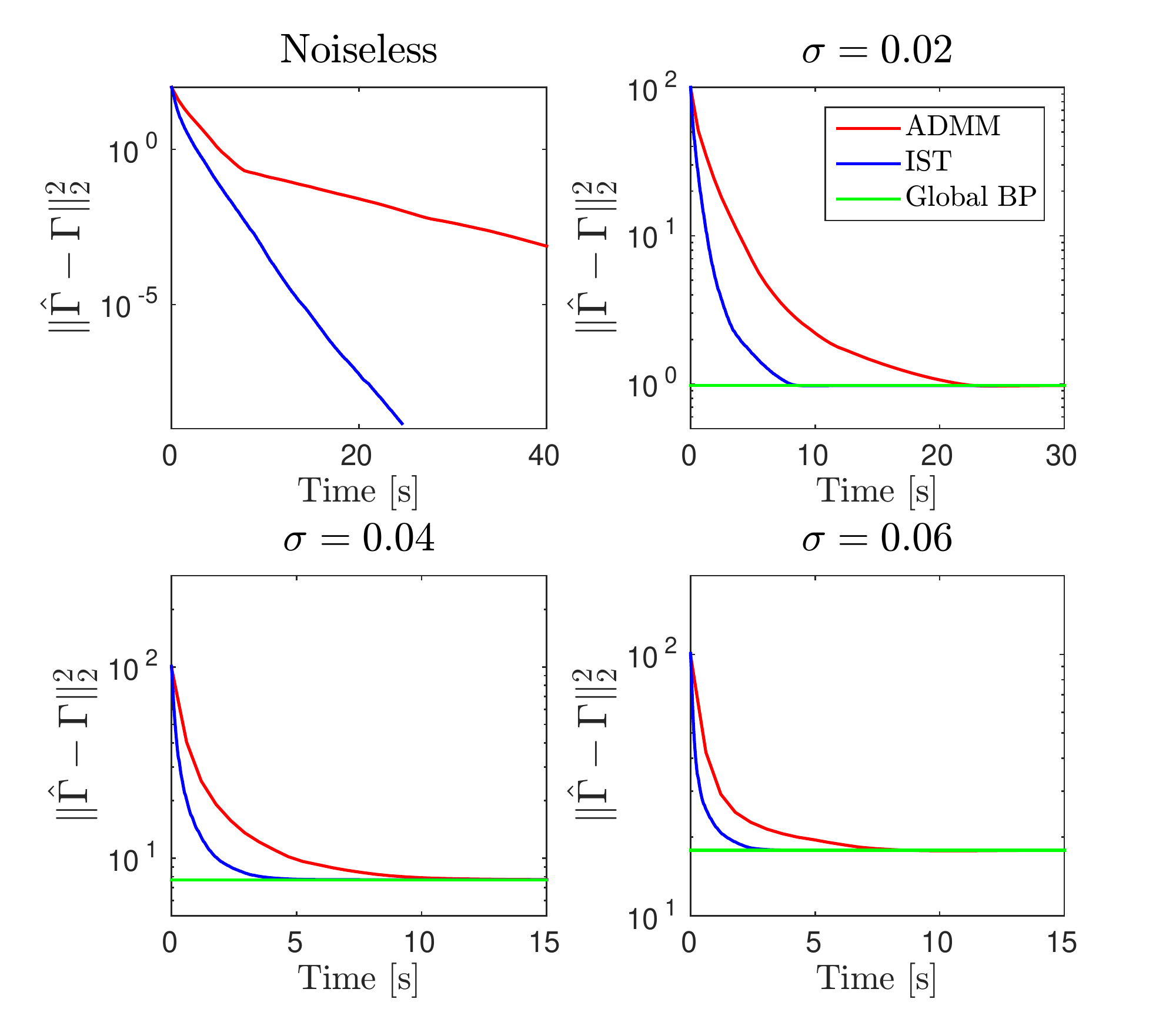}
	\caption{Distance between the estimate $\hat{\Gama}$ and the underlying solution $\Gama$ as a function of time for the IST and the ADMM algorithms compared to the solution obtained by solving the global BP.}
	\label{fig:ADMM_IT}
	\vspace{-0.3cm}
\end{figure}

We first employ the ADMM and IST algorithms in a noiseless scenario in order to minimize the global BP and find the underlying sparse vector. Since there is no noise added in this case, we decrease the penalty parameter $\lambda$ progressively throughout the iterations, making this value tend to zero as suggested in the previous subsection. In Figure \ref{fig:ADMM_L1_noiseless_codes} we present the evolution of the estimated $\hat{\Gama}$ for the ADMM solver throughout the iterations, after the global update stage. Note how the algorithm progressively increases the consensus and eventually recovers the true sparse vector. Equivalent plots are obtained for the IST method, and these are therefore omitted.

To extend the experiment to the noisy case, we contaminate the previous signal with additive white Gaussian noise of different standard deviations: $\sigma=0.02,0.04,0.06$. We then employ both local algorithms to solve the corresponding BPDN problems, and analyze the $\ell_2$ distance between their estimated sparse vector and the true one, as a function of time. These results are depicted in Figure \ref{fig:ADMM_IT}, where we include for completion the distance of the solution achieved by the global BP in the noisy cases. A few observations can be drawn from these results. Note that both algorithms converge to the solution of the global BP in all cases. In particular, the IST converges significantly faster than the ADMM method. Interestingly, despite the later requiring a smaller number of iterations to converge, these are relatively more expensive than those of the IST, which employs only multiplications by the small $\D_L$.

%The $\ell_2$ norm squared between these, $\|\hat{\Gama}-\Gama\|^2_2$, and the feasibility values are shown in Figure \ref{fig:ADMM_L1_noiseless_obj_feas}. Due to the bi-level consensus, we have both local and global feasibility terms, which correspond to the increments of the dual variables $\u_i$ and $\overline{\u}_i$, respectively. Formally, the first is given by $\sqrt{\sum_i \|\mathbf{Q} \gama_i - \alfa_i \|_2^2}$, whereas the later is defined as $\sqrt{\sum_j \|\mathbf{S}_j \Gama - \gama_j \|_2^2}$.
%
%Before concluding this section, we revisit the experiment carried for the ADMM algorithm. We run the IST on the same signal, minimizing the BP objective. We avoid plotting the obtained sparse codes as these look the same as the ones obtained by the ADMM. 
%Regarding the number of iterations, the IST requires considerably more iterations than the ADMM, and thus we run it for $50,000$ iterations. However, these iterations are a significantly cheaper due to the use of the small $\D_L$. To demonstrate this, we plot in Figure \ref{fig:ADMMvsIT} the $\|\hat{\Gama}-\Gama\|_2^2$ as a function of time for both algorithms. These implementations use un-optimized code, and improvement can probably be expected for both cases.

\vspace{-0.25cm}
\section{Conclusion and Future Work}
\label{Sec:Conclusions}

Striding on the foundations paved in the first part of this work, we have presented here a series of stability results for the convolutional sparse model in the presence of noise, providing guarantees for corresponding pursuit algorithms. These were possible due to our migration from the $\ell_0$ to the $\Loi$ norm, together with the generalization and utilization of concepts such as RIP and ERC. Seeking for a connection between traditional patch-based processing and the convolutional sparse model, we have proposed two efficient methods, that solve the global pursuit while working locally. 

We envision many possible directions of future work, and in what follows we present some of them:
\begin{itemize}
	\item We could extend our study, which considers only worst-case scenarios, to an average-performance analysis. By assuming more information about the model, it might be possible to quantify the probability of success of pursuit methods in the convolutional case. Such results would close the gap between current bounds and empirical results, as presented in both parts of this work.
	
	\item From an application point of view, we envision that interesting algorithms could be proposed to tackle real problems in signal and image processing while using the convolutional model. We note that while convolutional sparse coding has been applied to various problems, simple inverse problems such as denoising have not yet been properly addressed. We believe that the analysis presented in this work could facilitate the development of such algorithms by showing how to leverage on the subtleties of this model.	

	\item Interestingly, even though we have declared the $\Poi$ problem as our goal, at no point have we actually attempted to tackle it directly. What we have shown instead is that popular algorithms succeed in finding its solution. One could perhaps propose an algorithm specifically tailored for solving this problem -- or its convex relaxation ($\ell_{1,\infty}$). Such a method might be beneficial from both a theoretical and a practical aspect.
\end{itemize}
All these points, and more, are matter of current research.

\vspace{-0.25cm}
\section{Acknowledgements}
The research leading to these results has received funding from the European Research Council under European Union’s Seventh Framework Programme, ERC Grant agreement no. 320649. The authors would like to thank Dmitry Batenkov, Yaniv Romano and Raja Giryes for the prolific conversations and most useful advice which helped shape this work.

\appendices
\renewcommand{\theequation}{A-\arabic{equation}}
\setcounter{equation}{0}  % reset counter   
\section{OMP Stability Guarantee (Proof of Theorem \ref{Theorem:StabilityOMP})} \label{App:StabilityOMP}
\begin{proof}
	We shall first prove that the first step of OMP succeeds in recovering an element from the correct support.
	Denoting by $\mathcal{T}$ the support of $\Gama$, we can write
	\begin{equation}
	\Y = \D\Gama + \E = \sum_{t\in \mathcal{T}} \Gamma_t \d_t + \E.
	\label{GlobalExpression}
	\end{equation}
	Suppose that $\Gama$ has its largest coefficient in absolute value in $\Gamma_i$. For the first step of OMP to choose one of the atoms in the support, we require
	\begin{equation}
	|\d_i^T \Y | > \max_{j\notin\mathcal{T}} | \d_j^T \Y |.
	\end{equation}
	Substituting Equation \eqref{GlobalExpression} in this requirement we obtain
	\begin{equation} \label{eq:inequality}
	\left| \sum_{t\in \mathcal{T}}\Gamma_t\d_t^T\d_i + \E^T\d_i\right| > \max_{j\notin\mathcal{T}} \left| \sum_{t\in \mathcal{T}}\Gamma_t\d_t^T\d_j + \E^T\d_j\right|.
	\end{equation}
	Using the reverse triangle inequality we can construct a lower bound for the left hand side:
	\begin{align}
	\left| \sum_{t\in \mathcal{T}}\Gamma_t\d_t^T\d_i + \E^T\d_i\right|
	&\geq\left| \sum_{t\in \mathcal{T}}\Gamma_t\d_t^T\d_i\right| - \left|\E^T\d_i\right|.
	\end{align}
	Our next step is to bound the absolute value of the inner product of the noise and the atom $\d_i$. A na\"ive approach, based on the Cauchy-Schwarz inequality and the normalization of the atoms, would be to bound the inner product as $|\E^T\d_i| \leq \|\E\|_2\cdot\|\d_i\|_2 \leq \epsilon$. However, such bound would disregard the local nature of the atoms. Due to their limited support we have that $\d_i=\R_i^T\R_i\d_i$ where, as previously defined, $\R_i$ extracts a $n$-dimensional patch from a $N$-dimensional signal. Based on this observation, we have that
	\begin{equation}
	|\E^T\d_i| = |\E^T\R_i^T\R_i\d_i| \leq \|\R_i\E\|_2\cdot\|\d_i\|_2  \leq \epsilon_{_L},
	\end{equation}		
	where we have used the fact that \mbox{$\|\R_i \E\|_2 \leq \epsilon_{_L} \ \forall\ i$}. By exploiting the locality of the atom, together with the assumption regarding the maximal local energy of the noise, we are able to obtain a much tighter bound, because $\epsilon_{_L} \ll \epsilon$ in general. As a result, we obtain
	\begin{equation}
	\left| \sum_{t\in \mathcal{T}}\Gamma_t\d_t^T\d_i + \E^T\d_i\right| \geq\left| \sum_{t\in \mathcal{T}}\Gamma_t\d_t^T\d_i\right| - \epsilon_{_L}.
	\end{equation}
	Using the reverse triangle inequality, the normalization of the atoms and the fact that $|\Gamma_i|\geq|\Gamma_t|$, we obtain
	\begin{align}
	\left| \sum_{t\in \mathcal{T}}\Gamma_t\d_t^T\d_i + \E^T\d_i\right|
	&\geq |\Gamma_i| - \sum_{t \in \mathcal{T},t\neq i}|\Gamma_t|\cdot|\d_t^T\d_i |-\epsilon_{_L} \\
	&\geq |\Gamma_i| - |\Gamma_i|\sum_{t \in \mathcal{T},t\neq i}|\d_t^T\d_i |-\epsilon_{_L}.
	\end{align}
	Notice that $\d_t^T\d_i$ is zero for every atom too far from $\d_i$ because the atoms do not overlap. Denoting the stripe which fully contains the $i^{th}$ atom as $p(i)$ and its support as $\mathcal{T}_{p(i)}$, we can restrict the summation as:
	\begin{equation}
	\left| \sum_{t\in \mathcal{T}}\Gamma_t\d_t^T\d_i + \E^T\d_i\right| \geq |\Gamma_i| - |\Gamma_i|\sum_{\substack{t \in \mathcal{T}_{p(i)},\\ t\neq i}}|\d_t^T\d_i |-\epsilon_{_L}.
	\end{equation}
	Denoting by $n_{p(i)}$ the number of non-zeros in the support $\mathcal{T}_{p(i)}$ and using the definition of the mutual coherence we obtain:
	\begin{align*}
	\left| \sum_{t\in \mathcal{T}}\Gamma_t\d_t^T\d_i + \E^T\d_i\right| &	\geq |\Gamma_i| - |\Gamma_i| (n_{p(i)}-1)\mu(\D) - \epsilon_{_L}\\
	& \geq |\Gamma_i| - |\Gamma_i| (\|\Gama\|_{0,\infty}-1)\mu(\D) - \epsilon_{_L}.
	\end{align*}
	In the last inequality we have used the definition of the $\Loi$ norm.
	
	Now, we construct an upper bound for the right hand side of equation \eqref{eq:inequality}, once again using the triangle inequality and the fact that $|\E^T\d_j|\leq\epsilon_{_L}$:
	\begin{align}
	\max_{j\notin\mathcal{T}}\left| \sum_{t\in \mathcal{T}}\Gamma_t\d_t^T\d_j + \E^T\d_j \right|
	% & \leq \max_{j\notin\mathcal{T}}\left| \sum_{t\in \mathcal{T}}\Gamma_t\d_t^T\d_j\right| + \left|\E^T\d_j \right|\\
	& \leq \max_{j\notin\mathcal{T}}\left| \sum_{t\in \mathcal{T}}\Gamma_t\d_t^T\d_j\right| + \epsilon_{_L}.
	\end{align}
	Using the same rationale as before we get
	\begin{align}
	\max_{j\notin\mathcal{T}}\left| \sum_{t\in \mathcal{T}}\Gamma_t\d_t^T\d_j + \E^T\d_j \right| 
	&\leq |\Gamma_i|\max_{j\notin\mathcal{T}} \sum_{t \in \mathcal{T}}|\d_t^T\d_j |+\epsilon_{_L}\\
	&\leq |\Gamma_i|\max_{j\notin\mathcal{T}} \sum_{t \in \mathcal{T}_{p(j)}}|\d_t^T\d_j |+\epsilon_{_L}\\
	&\leq |\Gamma_i|\cdot\|\Gama\|_{0,\infty}\cdot\mu(\D)+\epsilon_{_L}.
	\end{align}
%	where, as before, $p(j)$ denotes the stripe which fully contains the $j^{th}$ atom. Once again, we have constrained the summation to elements which have non-zero inner product with the atom $\d_j$, those found inside the support $\mathcal{T}_{p(j)}$. Denoting by $n_{p(j)}$ the number of non-zeros in the support $\mathcal{T}_{p(j)}$ and using the definitions of the mutual coherence and the $\Loi$ norm we obtain:
%	\begin{align}
%	\max_{j\notin\mathcal{T}}\left| \sum_{t\in \mathcal{T}}\Gamma_t\d_t^T\d_j + \E^T\d_j \right|
%	&\leq |\Gamma_i|\max_{j\notin\mathcal{T}}\sum_{t \in \mathcal{T}_{p(j)}}|\d_t^T\d_j |+ \epsilon_{_L} \\
%	&\leq |\Gamma_i|\max_{j\notin\mathcal{T}}\ n_{p(j)}\cdot\mu(\D)+\epsilon_{_L} \\
%	&\leq |\Gamma_i|\cdot\|\Gama\|_{0,\infty}\cdot\mu(\D)+\epsilon_{_L}.
%	\end{align}
	Using both bounds, we obtain
	\begin{align*}
	& \left| \sum_{t\in \mathcal{T}}\Gamma_t\d_t^T\d_i + \E^T\d_i\right|
	\geq |\Gamma_i| - |\Gamma_i| (\|\Gama\|_{0,\infty}-1)\mu(\D) - \epsilon_{_L} \\
	\geq & |\Gamma_i|\cdot \|\Gama\|_{0,\infty}\mu(\D)+\epsilon_{_L}
	\geq\max_{j\notin\mathcal{T}}\left| \sum_{t\in \mathcal{T}}\Gamma_t\d_t^T\d_j + \E^T\d_j \right|.
	\end{align*}
	From this, it follows that
	\begin{align}
	\|\Gama\|_{0,\infty}&\leq\frac{1}{2}\left(1+\frac{1}{\mu(\D)}\right)-\frac{1}{\mu(\D)}\cdot\frac{\epsilon_{_L}}{|\Gamma_i|}.
	\label{Eq:ThisIneq}
	\end{align}
	Note that the theorem's hypothesis assumes that the above holds for $|\Gamma_{\text{min}}|$ instead of $|\Gamma_i|$. However, because $|\Gamma_i| \geq |\Gamma_{min}|$, this condition holds for every $i$. Therefore, Equation \eqref{Eq:ThisIneq} holds and we conclude that the first step of OMP succeeds.  
	
	Next, we address the success of subsequent iterations of the OMP.
	Define the sparse vector obtained after $k<\|\Gama\|_0$ iterations as $\Lamda^k$, and denote its support by $\mathcal{T}^k$. Assuming that the algorithm identified correct atoms (i.e., has so far succeeded), $\mathcal{T}^k = \text{supp}\{\Lamda_k\}\subset\text{supp}\{\Gama\}$. The next step in the algorithm is the update of the residual. This is done by decreasing a term proportional to the chosen atoms from the signal; i.e., 
	\begin{equation}
	\Y^k = \Y - \sum_{i\in \mathcal{T}^k} \d_i \Lamda_i^k.
	\end{equation}
	Moreover, $\Y^k$ can be seen as containing a clean signal $\X^k$ and the noise component $\E$, where 
	\begin{equation}
	\X^k = \X - \sum_{i\in \mathcal{T}^k} \d_i \Lamda_i^k = \D\Gama^k.
	\end{equation}
	The objective is then to recover the support of the sparse vector corresponding to $\X^k$, $\Gama^k$, defined as\footnote{Note that if $k=0$, $\X_0 = \X$, $\Y_0 = \Y$, and $\Gama_0 = \Gama$.}
	\begin{equation}  \label{eq:defGama_k}
	\Gamma_i^k =
	\left\{
	\begin{array}{ll}
	\Gamma_i - \Lamda_i^k & \mbox{if } \ i\in \mathcal{T}^k \\
	\Gamma_i  & \mbox{if } \ i\notin \mathcal{T}^k.
	\end{array}
	\right.
	\end{equation}	
	Note that $\text{supp}\{\Gama^k\} \subseteq \text{supp}\{\Gama\}$ and so
	\begin{equation} \label{eq:omp_first_item}
	\|\Gama^k\|_{0,\infty}\leq \|\Gama\|_{0,\infty}.
	\end{equation}
	In words, the $\Loi$ norm of the underlying solution of $\X^k$ does not increase as the iterations proceed. Note that this representation is also unique in light of the uniqueness theorem presented in part I.
	From the above definitions, we have that
	\begin{align}
	\Y^k - \X^k &= \Y - \sum_{i\in \mathcal{T}^k} \d_i \Lamda_i^k - \X + \sum_{i\in \mathcal{T}^k} \d_i \Lamda_i^k \\
	&= \Y - \X = \E.
	\end{align}
	Hence, the noise level is preserved, both locally and globally; both $\epsilon$ and $\epsilon_L$ remain the same. 
	
	Note that $\Gama^k$ differs from $\Gama$ in at most $k$ places, following Equation \eqref{eq:defGama_k} and that $|\mathcal{T}^k|=k$. As such, $\|\Gama^k\|_{\infty}$ is greater than the $(k+1)^{th}$ largest element in absolute value in $\Gama$. This implies that $\|\Gama^k\|_{\infty}\geq |\Gamma_{\min}|$.
	Finally, we obtain that
	\begin{align}
	\|\Gama^k\|_{0,\infty} \leq \|\Gama\|_{0,\infty}
	& < \frac{1}{2}\left( 1+\frac{1}{\mu(\D)} \right)-\frac{1}{\mu(\D)}\cdot\frac{\epsilon_{_L}}{|\Gamma_{min}|} \\
	& \leq \frac{1}{2}\left( 1+\frac{1}{\mu(\D)} \right)-\frac{1}{\mu(\D)}\cdot\frac{\epsilon_{_L}}{\|\Gama^k\|_\infty}.
	\end{align}
	The first inequality is due to \eqref{eq:omp_first_item}, the second is the assumption in \eqref{omp_hypothesis} and the third was just obtained above. Thus,
	\begin{equation}
	\|\Gama^k\|_{0,\infty} < \frac{1}{2}\left( 1+\frac{1}{\mu(\D)} \right)-\frac{1}{\mu(\D)}\cdot\frac{\epsilon_{_L}}{\|\Gama^k\|_\infty}.
	\end{equation}		
	Similar to the first iteration, the above inequality together with the fact that the noise level is preserved, guarantees the success of the next iteration of the OMP algorithm. From this follows that the algorithm is guaranteed to recover the true support after $\|\Gama\|_0$ iterations.
	
	Finally, we move to prove the second claim. In its last iteration OMP solves the following problem:
	\begin{equation}
	\Gama_{OMP}=\arg\min_\Delt\|\D_{\mathcal{T}}\Delt-\Y\|_2^2,
	\end{equation}
	where $\D_{\mathcal{T}}$ is the convolutional dictionary restricted to the support $\mathcal{T}$ of the true sparse code $\Gama$. Denoting $\Gama_\mathcal{T}$ the (dense) vector corresponding to those atoms, the solution to the above problem is simply given by
	\begin{align}
	\Gama_{OMP}
	&=\D_\mathcal{T}^{\dagger}\Y
	=\D_\mathcal{T}^{\dagger}\left( \D\Gama + \E \right) \\
	&=\D_\mathcal{T}^{\dagger}\left( \D_\mathcal{T}\Gama_\mathcal{T} + \E \right)
	=\Gama_\mathcal{T} + \D_\mathcal{T}^{\dagger}\E,
	\end{align}
	where we have denoted by $\D_\mathcal{T}^{\dagger}$ the Moore-Penrose pseudoinverse of the sub-dictionary $\D_\mathcal{T}$.
	Thus,
	\begin{align}
	\|\Gama_{OMP}-\Gama_\mathcal{T}\|_2^2 = \|\D_\mathcal{T}^{\dagger}\E\|_2^2
	&\leq \|\D_\mathcal{T}^{\dagger}\|^2_2\cdot\|\E\|^2_2 \\
	= \frac{1}{\lambda_{\min}\left(\D_\mathcal{T}^T\D_\mathcal{T} \right)}\|\E\|^2_2
	&\leq\frac{\epsilon^2}{1-\mu(\D)(\|\Gama\|_{0,\infty}-1)}.
	\end{align}
	In the last inequality we have used the bound on the eigenvalues of $\D_\mathcal{T}^T\D_\mathcal{T}$ derived in Lemma 1, in part I.
	
\end{proof}

\renewcommand{\theequation}{B-\arabic{equation}}
\setcounter{equation}{0}  % reset counter  
\vspace{-.75cm}
\section{ERC in the Convolutional Sparse Model \\ (Proof of Theorem \ref{Theorem:ERC_Loi})}
\label{App:ERC_Loi}
\begin{proof}
	For the ERC to be satisfied, we must require that, for every $i \notin \mathcal{T}$, 
	\begin{equation}
	\| \D^{\dagger}_\mathcal{T} \d_i \|_1  =  \left\| \left( \D^T_\mathcal{T} \D_\mathcal{T} \right)^{-1} \D^T_\mathcal{T} \d_i \right\|_1 < 1.
	\end{equation}
	Using properties of induced norms, we have that
	\begin{equation} \label{eq:ERC_multiplicative}
	\left\| \left( \D^T_\mathcal{T} \D_\mathcal{T} \right)^{-1} \D^T_\mathcal{T} \d_i \right\|_1
	\leq \left\| \left( \D^T_\mathcal{T} \D_\mathcal{T} \right)^{-1} \right\|_1 \left\| \D^T_\mathcal{T} \d_i \right\|_1.
	\end{equation}
	Using the definition of the mutual coherence, it is easy to see that the absolute value of the entries in the vector $\D^T_\mathcal{T} \d_i$ are bounded by $\mu(\D)$. Moreover, due to the locality of the atoms, the number of non-zero inner products with the atom $\d_i$ is equal to the number of atoms in $\mathcal{T}$ that overlap with it. This number can, in turn, be bounded by the maximal number of non-zeros in a stripe from $\mathcal{T}$, i.e., its $\Loi$ norm, denoted by $k$. Therefore, $\left\| \D^T_\mathcal{T} \d_i \right\|_1 \leq k \mu(\D)$. 
	
	Addressing now the first term in Equation \eqref{eq:ERC_multiplicative}, note that
	\begin{equation} \label{eq:ERC_gram_l1_bound}
	\left\| \left( \D^T_\mathcal{T} \D_\mathcal{T} \right)^{-1} \right\|_1 = \left\| \left( \D^T_\mathcal{T} \D_\mathcal{T} \right)^{-1} \right\|_\infty,
	\end{equation}
	since the induced $\ell_1$ and $\ell_\infty$ norms are equal for symmetric matrices. Next, using the Ahlberg-Nilson-Varah bound and similar steps to those presented in Lemma 1, we have that
	\begin{equation} \label{eq:ERC_gram_infinity_bound}
	\left\| \left( \D^T_\mathcal{T} \D_\mathcal{T} \right)^{-1} \right\|_\infty \leq \frac{1}{1 - (k-1) \mu(\D) }.
	\end{equation}
	In order for this to hold, we must require the Gram $\D^T_\mathcal{T} \D_\mathcal{T}$ to be diagonally dominant, which is satisfied if $1-(k-1)\mu(\D) > 0$. This is indeed the case, as follows from the assumption on the $\Loi$ norm of $\mathcal{T}$. Plugging the above into Equation \eqref{eq:ERC_multiplicative}, we obtain
	\begin{align} \label{eq:ERC_theta_bound}
	\left\| \left( \D^T_\mathcal{T} \D_\mathcal{T} \right)^{-1} \D^T_\mathcal{T} \d_i \right\|_1
	& \leq \left\| \left( \D^T_\mathcal{T} \D_\mathcal{T} \right)^{-1} \right\|_1 \left\| \D^T_\mathcal{T} \d_i \right\|_1 \\
	& \leq \frac{k \mu(\D)}{1 - (k-1) \mu(\D) }.
	\end{align}
	Our assumption that $k < \frac{1}{2}\left(1+\frac{1}{\mu(\D)}\right)$ implies that the above term is less than one, thus showing the ERC is satisfied for all supports $\mathcal{T}$ that satisfy $\|\mathcal{T}\|_{0,\infty} < \frac{1}{2}\left(1+\frac{1}{\mu(\D)}\right)$.
\end{proof}

\renewcommand{\theequation}{C-\arabic{equation}}
\setcounter{equation}{0}  % reset counter  
\section{Basis Pursuit Stability Guarantee (Proof of Theorem \ref{Theorem:StabilityBP})}
\label{App:StabilityBP}

We first state and prove a Lemma that will become of use while proving the stability result of BP.

\begin{lemma}{}
\label{lemma:noise_correlation}
Suppose a clean signal $\X$ has a representation $\D\Gama$, and that it is contaminated with noise $\E$ to create the signal $\Y=\X+\E$. Denote by $\epsilon_{_L}$ the highest energy of all $n$-dimensional local patches extracted from $\E$. Assume that
\begin{equation} \label{eq:lemma_assumption}
\|\Gama\|_{0,\infty}<\frac{1}{2} \left( 1 + \frac{1}{\mu(\D)} \right).
\end{equation}
Denoting by $\X_{\text{LS}}$ the best $\ell_2$ approximation of $\Y$ over the support $\mathcal{T}$, we have that\footnote{We suspect that, perhaps under further assumptions, the constant in this bound can be improved from 2 to 1. This is motivated by the fact that the bound in \cite{Tropp2006}, for the traditional sparse model, is $1\cdot \epsilon$ -- where $\epsilon$ is the global noise level.}
\begin{equation} \label{eq:DT_res}
\| \D^T(\Y-\X_{\text{LS}}) \|_\infty \leq 2\epsilon_L.
\end{equation}
\end{lemma}

\begin{proof}
Using the expression for the least squares solution (and assuming that $\DT$ has full-column rank), we have that
\begin{align*}
\DT^T(\Y-\X_{\text{LS}})
& = \DT^T \left( \Y-\DT \left( \DT^T \DT \right)^{-1} \DT^T \Y \right) \\
& = \left( \DT^T - \DT^T \DT \left( \DT^T \DT \right)^{-1} \DT^T \right) \Y = \mathbf{0}.
\end{align*}
%The $i^{th}$ entry in the vector $\D^T(\Y-\X_{\text{LS}})$ is equal to the inner product of the atom $\d_i$ with the vector $\Y-\X_\text{LS}$. 
This shows that all inner products between atoms inside $\mathcal{T}$ and the vector $\Y-\X_{\text{LS}}$ are zero, and thus $\| \DTB^T(\Y-\X_{\text{LS}}) \|_\infty = \| \D^T(\Y-\X_{\text{LS}}) \|_\infty$. We have denoted by $\overline{\mathcal{T}}$ the complement to the support, containing all atoms not found in $\mathcal{T}$, and by $\DTB$ the corresponding dictionary.
Denoting by $\Gama_\mathcal{T}$ the vector $\Gama$ restricted to its support, and expressing $\X_{\text{LS}}$ and  $\Y$ conveniently, we obtain
\begin{align}
& \| \DTB^T(\Y-\X_{\text{LS}}) \|_\infty \\
= &\| \DTB^T \left( \mathbf{I} - \DT \left( \DT^T \DT \right)^{-1} \DT^T \right) \Y \|_\infty \\
= &\| \DTB^T \left( \mathbf{I} - \DT \left( \DT^T \DT \right)^{-1} \DT^T \right) (\DT\Gama_\mathcal{T} + \E) \|_\infty. \label{eq:DTB_res}
\end{align}
It is easy to verify that 
\begin{align}
\phantom{=} \left( \mathbf{I} - \DT \left( \DT^T \DT \right)^{-1} \DT^T \right) \DT\Gama_\mathcal{T} = \mathbf{0}.
\end{align}
Plugging this into the above, we have that 
\begin{align*}
 \| \DTB^T(\Y-\X_{\text{LS}}) \|_\infty  =  \left\| \DTB^T \left( \mathbf{I} - \DT \left( \DT^T \DT \right)^{-1} \DT^T \right) \E \right\|_\infty.
%= & \left\| \DTB^T \left( \mathbf{I} - \DT \left( \DT^T \DT \right)^{-1} \DT^T \right) (\DT\Gama_\mathcal{T} + \E) \right\|_\infty \\
\end{align*}
Using the triangle inequality for the $\ell_\infty$ norm, we obtain
\begin{align}
& \| \DTB^T(\Y-\X_{\text{LS}}) \|_\infty \\
= & \left\| \DTB^T\E - \DTB^T\DT \left( \DT^T \DT \right)^{-1} \DT^T\E \right\|_\infty \\
\leq & \left\| \DTB^T\E \right\|_\infty + \left\| \DTB^T\DT \left( \DT^T \DT \right)^{-1} \DT^T\E \right\|_\infty. \label{eq:triangle_inequality_step}
\end{align}
In what follows, we will bound both terms in the above expression with $\epsilon_L$. First, due to the limited support of the atoms, $\d_i=\R_i^T\R_i\d_i$, where $\R_i$ extracts the $i^{th}$ local patch from the global signal, as previously defined. Thus,
\begin{align} \label{eq:local_noise}
\left\| \DTB^T\E \right\|_\infty
 = \max_{i\in\overline{\mathcal{T}}} | \d_i^T\E | & = \max_{i\in\overline{\mathcal{T}}} | \d_i^T\R_i^T\R_i\E | \\
& \leq \max_{i\in\overline{\mathcal{T}}} \|\R_i\d_i\|_2 \cdot \|\R_i\E\|_2 \leq \epsilon_L,
\end{align}
where we have used the Cauchy-Schwarz inequality, the normalization of the atoms and the fact that $\|\R_i\E\|_2\leq\epsilon_L$ $\forall \ i$. Next, we move to the second term in Equation \eqref{eq:triangle_inequality_step}. Using the definition of the induced $\ell_\infty$ norm, and the bound $\| \DT^T\E \|_\infty \leq \epsilon_L$, we have that
\begin{align*}
\left\| \DTB^T\DT \left( \DT^T \DT \right)^{-1} \DT^T\E \right\|_\infty \leq &\left\| \DTB^T\DT \left( \DT^T \DT \right)^{-1} \right\|_\infty \epsilon_L.
\end{align*}
Recall that the induced infinity norm of a matrix is equal to the maximal $\ell_1$ norm of its rows. Notice that a row in the above matrix can be written as $\d_i^T\DT \left( \DT^T \DT \right)^{-1}$, where $i\in \overline{\mathcal{T}}$. Then,
\begin{align}
&\left\| \DTB^T\DT \left( \DT^T \DT \right)^{-1} \DT^T\E \right\|_\infty \\
\leq &\max_{i\in\overline{\mathcal{T}}} \left\| \d_i^T\DT \left( \DT^T \DT \right)^{-1} \right\|_1 \cdot \epsilon_L.
\end{align}
Using the definition of induced $\ell_1$ norm and Equation \eqref{eq:ERC_gram_l1_bound} and \eqref{eq:ERC_gram_infinity_bound}, we obtain that
\begin{align}
&\left\| \DTB^T\DT \left( \DT^T \DT \right)^{-1} \DT^T\E \right\|_\infty \\
\leq &\max_{i\in\overline{\mathcal{T}}} \left\| \d_i^T\DT \right\|_1 \cdot \left\| \left( \DT^T \DT \right)^{-1} \right\|_1 \cdot \epsilon_L \\
\leq &\max_{i\in\overline{\mathcal{T}}} \left\| \d_i^T\DT \right\|_1 \cdot \frac{1}{1-(k-1)\mu(\D)} \cdot \epsilon_L,
\end{align}
where we have denoted by $k$ the $\Loi$ norm of $\mathcal{T}$. Notice that due to the limited support of the atoms, the vector $\d_i^T\DT$ has at most $k$ non-zeros entries. Additionally, each of these is bounded in absolute value by the mutual coherence of the dictionary. Therefore, $\|\d_i^T\DT\|_1 \leq k\mu(\D)$ (note that $i\notin\mathcal{T}$). Plugging this into the above equation, we obtain
\begin{align}
\left\| \DTB^T\DT \left( \DT^T \DT \right)^{-1} \DT^T\E \right\|_\infty \leq \frac{k\mu(\D)}{1-(k-1)\mu(\D)} \cdot \epsilon_L.
\end{align}
Rearranging our assumption in Equation \eqref{eq:lemma_assumption}, we get $\frac{k\mu(\D)}{1-(k-1)\mu(\D)}<1$. Therefore, the above becomes
\begin{align} \label{eq:using_assumption}
\left\| \DTB^T\DT \left( \DT^T \DT \right)^{-1} \DT^T\E \right\|_\infty < \epsilon_L.
\end{align}
Finally, plugging Equation \eqref{eq:local_noise} and \eqref{eq:using_assumption} into Equation \eqref{eq:triangle_inequality_step}, we conclude that
\begin{align}
& \| \DTB^T(\Y-\X_{\text{LS}}) \|_\infty \\
\leq & \left\| \DTB^T\E \right\|_\infty + \left\| \DTB^T\DT \left( \DT^T \DT \right)^{-1} \DT^T\E \right\|_\infty \\
\leq & \ \epsilon_L + \epsilon_L = 2\epsilon_L.
\end{align}
\vspace{-0.2cm}
\end{proof}

For completeness, and before moving to the proof of the stability of BP, we now reproduce Theorem 8 from \cite{Tropp2006}.
\begin{thm}{(Tropp):} \label{Theorem:Tropp}
	Suppose a clean signal $\X$ has a representation $\D\Gama$, and that it is contaminated with noise $\E$ to create the signal $\Y=\X+\E$. Assume further that $\Y$ is a signal whose best $\ell_2$ approximation over the support of $\Gama$, denoted by $\mathcal{T}$, is given by $\X_{\text{LS}}$, and that $\X_{\text{LS}} = \D \Gama_{\text{LS}}$. Moreover, consider $\Gama_{\text{BP}}$ to be the solution to the Lagrangian BP formulation (as in Equation \eqref{eq:lagrangian_BP}) with parameter $\lambda$. If the following conditions are satisfied:
	\begin{enumerate}[ a) ]
		\item \label{eq:ERC_assumption} The ERC is met with constant $\theta$ for the support $\mathcal{T}$; And
		\item \label{eq:Corr_assumption} $\| \D^T(\Y-\X_{\text{LS}}) \|_\infty \leq \lambda\theta$,
	\end{enumerate}
	then the following hold:
	\begin{enumerate}
		\item The support of $\Gama_{\text{BP}}$ is contained in that of $\Gama$.
		\item \label{thesis_2} $\|\Gama_{\text{BP}}-\Gama_{\text{LS}} \|_\infty < \lambda \left\| \left( \D_\mathcal{T}^T \D_\mathcal{T} \right)^{-1} \right\|_\infty$.
		\item \label{thesis_3} In particular, the support of $\Gama_{\text{BP}}$ contains every index $i$ for which $|{{\Gama_{\text{LS}}}_{ }}_i| > \lambda \left\| \left( \D_\mathcal{T}^T \D_\mathcal{T} \right)^{-1} \right\|_\infty$.
		\item The minimizer of the problem, $\Gama_{\text{BP}}$, is unique.
	\end{enumerate}
\end{thm}

Armed with these, we now proceed to the main concern of this section, proving Theorem \ref{Theorem:StabilityBP}.

\begin{proof}
In this proof we shall show that Theorem \ref{Theorem:Tropp} can be reformulated in terms of the $\Loi$ norm and the mutual coherence of $\D$, thus adapting it to the convolutional setting. Our strategy will be first to restrict its conditions \eqref{eq:ERC_assumption} and \eqref{eq:Corr_assumption}, and then to derive from its theses the desired claims.

To this end, we begin by converting the assumption on the ERC into another one relying on the $\Loi$ norm. This can be readily done using Theorem \ref{Theorem:ERC_Loi}, which states that the ERC is met assuming the $\Loi$ norm of the support is less than $\frac{1}{2} \left( 1 + \frac{1}{\mu(\D)} \right)$ -- a condition that is indeed satisfied due to our assumption in Equation \eqref{eq:BP_assumption}.
Next, we move to assumption \eqref{eq:Corr_assumption} in Theorem \ref{Theorem:Tropp}. We can lower bound the ERC constant $\theta$ by employing the inequality in \eqref{eq:ERC_theta_bound}, thus obtaining
\begin{equation} \label{eq:theta_inequality}
\theta =  1 - \underset{i \notin \mathcal{T}}{\max} \|\D^{\dagger}_{\mathcal{T}} \d_i \|_1 \geq 1 - \frac{\|\Gama\|_{0,\infty} \mu(\D)}{1 - (\|\Gama\|_{0,\infty}-1) \mu(\D) }.
\end{equation}
Using the assumption that $\|\Gama\|_{0,\infty}<\frac{1}{3} \left( 1 + \frac{1}{\mu(\D)} \right)$, as stated in Equation \eqref{eq:BP_assumption}, the above can be simplified into
\begin{equation} \label{eq:lower_bound_theta}
\theta =  1 - \underset{i \notin \mathcal{T}}{\max} \|\D^{\dagger}_{\mathcal{T}} \d_i \|_1 > \frac{1}{2}.
\end{equation}
Bringing now the fact that $\lambda = 4\epsilon_L$, as assumed in our Theorem, and using the just obtained inequality \eqref{eq:lower_bound_theta}, condition \eqref{eq:Corr_assumption} must hold since
\begin{equation}
\| \D^T(\Y-\X_{\text{LS}}) \|_\infty \leq 2 \epsilon_L < \theta  \lambda.
\end{equation}
Note that the leftmost inequality is Lemma \eqref{lemma:noise_correlation}, and the implication here is that $\lambda \geq 4 \epsilon_L$.

Thus far, we have addressed the conditions in Theorem \ref{Theorem:Tropp}, showing that they hold in our convolutional setting. In the remainder of this proof we shall expand on its results, in particular point 2 and 3. We can upper bound the term $\left\| \left( \D_\mathcal{T}^T \D_\mathcal{T} \right)^{-1} \right\|_\infty$ using Equation \eqref{eq:ERC_gram_infinity_bound}, obtaining
\begin{equation}
\left\| \left( \D_\mathcal{T}^T \D_\mathcal{T} \right)^{-1} \right\|_\infty \leq \frac{1}{1 - ( \|\Gama\|_{0,\infty} -1) \mu(\D) }.
\label{eq:inequationFinal}
\end{equation}
Using once again the assumption that $\|\Gama\|_{0,\infty}<\frac{1}{3}\left(1+\frac{1}{\mu(\D)}\right)$, we have that $\|\Gama\|_{0,\infty}<\frac{1}{3}\left(3+\frac{1}{\mu(\D)}\right)$. From this last inequality, we get $\left( \|\Gama\|_{0,\infty} - 1 \right) \mu(\D) < \frac{1}{3}$. Thus, it follows that
\begin{equation}
\frac{1}{1 - ( \|\Gama\|_{0,\infty} -1) \mu(\D) } < \frac{3}{2}.
\end{equation}
Based on the above inequality, and Equation \eqref{eq:inequationFinal}, we get
\begin{equation}\label{eq:InverseGramBound}
\left\| \left( \D_\mathcal{T}^T \D_\mathcal{T} \right)^{-1} \right\|_\infty < \frac{3}{2}.
\end{equation}
Plugging this into the second result in Tropp's theorem, together with the above fixed $\lambda$, we obtain that
\begin{equation} \label{eq:BoundGamaBPLS}
\|\Gama_{\text{BP}}-\Gama_{\text{LS}} \|_\infty < \lambda \left\| \left( \D_\mathcal{T}^T \D_\mathcal{T} \right)^{-1} \right\|_\infty < 4\epsilon_L \cdot \frac{3}{2} = 6\epsilon_L.
\end{equation}
On the other hand, looking at the distance to the real $\Gama$,
\begin{align} \label{eq:BoundLsLast}
	\|\Gama_{\text{LS}}-\Gama\|_{\infty} & = \| \left( \D_\mathcal{T}^T \D_\mathcal{T} \right)^{-1} \D_\mathcal{T}^T\left( \Y - \X \right)\|_{\infty}\\ 
									    & \leq  \| \left( \D_\mathcal{T}^T \D_\mathcal{T} \right)^{-1} \|_{\infty}\cdot \| \D_\mathcal{T}^T \E \|_{\infty} < \frac{3}{2} \epsilon_L.
\end{align}
For the first inequality we have used the definition of the induced $\ell_\infty$ norm, and the second one follows from \eqref{eq:InverseGramBound} and a similar derivation to that in \eqref{eq:local_noise}. Finally, using triangle inequality and Equations \eqref{eq:BoundLsLast} and \eqref{eq:BoundGamaBPLS} we obtain
\begin{equation}
\| \Gama_{\text{BP}} - \Gama \|_{\infty} \leq \| \Gama_{\text{BP}} - \Gama_{\text{LS}} \|_{\infty} + \| \Gama_{\text{LS}} - \Gama \|_{\infty} < \frac{15}{2} \epsilon_L.
\end{equation}
The third result in the theorem follows immediately from the above.

\end{proof}

\vspace{-1cm}
\bibliographystyle{ieeetr}
\bibliography{MyBib}

\end{document}